\documentclass[a4paper,UKenglish,cleveref, autoref, thm-restate]{lipics-v2019}

\usepackage{microtype}
\usepackage{xspace}

\title{Regular resynchronizability of origin transducers is undecidable}

\author{Denis Kuperberg}{CNRS, LIP, ENS Lyon, France}{denis.kuperberg@ens-lyon.fr}{https://orcid.org/0000-0001-5406-717X}{}

\author{Jan Martens}{Eindhoven University of Technology, Netherlands}{j.j.m.martens@tue.nl}{}{}

\authorrunning{D. Kuperberg and J. Martens}

\Copyright{Denis Kuperberg and Jan Martens}



\relatedversion{https://arxiv.org/abs/2002.07558}

\supplement{}

\funding{}

\nolinenumbers 

\EventEditors{Javier Esparza and Daniel Kr{\'a}l'}
\EventNoEds{2}
\EventLongTitle{45th International Symposium on Mathematical Foundations of Computer Science (MFCS 2020)}
\EventShortTitle{MFCS 2020}
\EventAcronym{MFCS}
\EventYear{2020}
\EventDate{August 24--28, 2020}
\EventLocation{Prague, Czech Republic}
\EventLogo{}
\SeriesVolume{170}
\ArticleNo{51}

\usepackage{hhline}
\usepackage{stmaryrd}
\usepackage{multirow}

\usepackage{tikz}
\usetikzlibrary{automata,positioning,decorations.pathreplacing}
\usepackage{xcolor}
\usepackage[edges]{forest}
\forestset{
	colour my roots/.style={
		before typesetting nodes={
			edge=#1,
		}
	},
}

\tikzstyle{before}=[->,blue,>=stealth]
\tikzstyle{after}=[->,dashed,red,>=stealth]

\newcommand{\ProofSketch}{\begin{proof}[Proof sketch.]}

\newcommand{\N}{\mathbb N}
\newcommand{\B}{\mathbb B}

\newcommand{\tuple}[1]{\langle #1 \rangle}
\newcommand{\dom}{\mathit{dom}}

\makeatletter
\newcommand{\mypm}{\mathbin{\mathpalette\@mypm\relax}}
\newcommand{\@mypm}[2]{\ooalign{%
  \raisebox{.1\height}{$#1+$}\cr
  \smash{\raisebox{-.6\height}{$#1-$}}\cr}}
\makeatother

\newcommand{\sem}[1]{\llbracket #1 \rrbracket}
\newcommand{\osem}[1]{\sem{#1}_o}
\newcommand{\Runiv}{R_\mathit{univ}}
\newcommand{\Rshift}{R_\mathit{\mypm 1}}
\newcommand{\Rblock}{R_\mathit{block}}

\newcommand{\Rfirstlast}{R_\mathit{1st-to-last}}

\newcommand{\A}{\mathcal A}
\newcommand{\Tlast}{T_\mathit{first}}
\newcommand{\Tfirst}{T_\mathit{last}}
\newcommand{\Tid}{T_\mathit{id}}
\newcommand{\Trev}{T_\mathit{rev}}
\newcommand{\Ta}{T_\mathit{one-two}}
\newcommand{\Taa}{T_\mathit{two-one}}
\newcommand{\Tslow}{T_\mathit{slow}}
\newcommand{\Tfast}{T_\mathit{fast}}
\newcommand{\Rreg}{R_{reg}}
\newcommand{\Rrat}{R_{rat}}

\newcommand{\Tup}{T_\mathit{up}}
\newcommand{\Tdown}{T_\mathit{down}}

\newcommand{\trans}[1]{\stackrel{#1}{\longrightarrow}}
\newcommand{\eps}{\varepsilon}
\newcommand{\toleft}{\mathsf{left}}
\newcommand{\toright}{\mathsf{right}}
\newcommand{\orig}{\mathrm{orig}}

\newcommand{\boundtape}{\mathit{BoundTape}}
\newcommand{\pref}{\sqsubseteq}
\newcommand{\hist}{\mathit{Hist}_M}

\newcommand{\pspace}{{\sc{PSpace}}\xspace}

\newcommand{\pfail}{p_\mathit{fail}}
\begin{document}

\maketitle

\begin{abstract}
We study the relation of containment up to unknown regular resynchronization between two-way non-deterministic transducers. We show that it constitutes a preorder, and that the corresponding equivalence relation is properly intermediate between origin equivalence and classical equivalence. We give a syntactical characterization for containment of two transducers up to resynchronization, and use it to show that this containment relation is undecidable already for one-way non-deterministic transducers, and for simple classes of resynchronizations. This answers the open problem stated in recent works, asking whether this relation is decidable for two-way non-deterministic transducers.
 \end{abstract}

\newpage

\section{Introduction}

The study of transductions, that is functions and relations from words to words, is a fundamental field of theoretical computer science. Many models of transducers have been proposed, and robust notions such as \emph{regular transductions} emerged \cite{engelfriet2001mso,alur2010expressiveness}.
However, many natural problems on transductions are undecidable, for instance equivalence of one-way non-deterministic transducers \cite{Griffiths68,ibarra1978unsolvability}.


In order to circumvent this, and to obtain a better-behaved model, Boja\'nczyk introduced transducers with \emph{origin information} \cite{bojanczyk2014transducers}, where the semantics takes into account not only the input/output pair of words, but also the way the output is produced from the input.
It is shown in \cite{bojanczyk2014transducers} that translations between different models of transducers usually preserve the origin semantics, more problems become decidable, such as the equivalence between two transducers, and the model of transduction with origins is more amenable to an algebraic approach.

The fact that two transducers are origin-equivalent if they produce their output in exactly the same way can seem too strict, and prompted the idea of \emph{resynchronization}.
The idea, introduced in \cite{filiot2016equivalence}, where the main focus was the \emph{sequential uniformization} problem, and developed in \cite{bose2018origin,bose2019synthesis}, is to allow a distortion of the origins in a controlled way, in order to recognize that two transducers have a similar behaviour.

It is shown in \cite{bose2018origin}, that containment of 2-way transducers up to a fixed resynchronization is in \pspace, so no more difficult than classical containment of non-deterministic one-way automata. This covers in particular the case where the resynchronization is trivial, in which case the problem boils down to testing strict origin equivalence.

In \cite{bose2019synthesis}, the \emph{resynchronizer synthesis} problem was studied. The goal is now to decide whether there exists a resynchronizer $R$ such that containment or equivalence holds up to $R$.
Some results are obtained for two notions of resynchronizers. The first notion, introduced in \cite{filiot2016equivalence} is called \emph{rational resynchronizers}, it is specialized for $1$-way transducers, and uses an interleaving of input and output letters. The second notion is called \emph{(bounded) regular resynchronizers}, it is the focus of \cite{bose2018origin} and is defined for two-way transducers.

For rational resynchronizers, a complete picture is obtained in \cite{bose2019synthesis}: the synthesis problem is decidable for $k$-valued transducers, but undecidable in general.
For regular resynchronizers, it is shown in \cite{bose2019synthesis} that the synthesis problem is decidable for unambiguous two-way transducers, i.e. transducers that have at most one accepting run on each input word. The ambiguous case is left open. It was also shown in \cite{bose2019synthesis} that for one-way transducers, the notion of rational and regular resynchronizer do not match.
The picture for resynchronizability from previous works is summed up in this table, where the first line describes constraints on the input pair of transducers:
\vspace{-.4cm}
\begin{center}
\begin{tabular}{|c|c|c|c|}
\cline{2-4}
\multicolumn{1}{c|}{} & unambiguous & functional/finite-valued & general case \\ 
\hline
Fixed regular resync. (2-way) & \pspace & \pspace-c & \pspace-c.\\
\hline
Unknown rational resync. (1-way) & decidable & decidable & undecidable\\
\hline
Unknown regular resync. (2-way) & decidable & ? & ?\\
\hline
\end{tabular}
\end{center}
%
In this work, we tackle the general case (last question mark), and show a stronger result: the synthesis of regular resynchronizers is already undecidable for one-way transducers.

To do so, we introduce the notion of limited traversal, which characterizes whether two transducers verify a containment relation up to some unknown resynchronization. 
Outside of this undecidability proof, this notion can be used to show that some natural transducers, equivalent in the classical sense, cannot be resynchronized.
As a by-product, we also obtain that the resynchronizer synthesis problem is undecidable even if we restrict regular resynchronizers to any natural subclass containing the simple ``shifting'' resynchronizations, allowing origins to change by at most $k$ positions for a fixed bound $k$.
Our proof can also be lifted to show a different statement, emphasizing the difference between rational and regular resynchronization: even in presence of regular resynchronization, synthesis of a rational resynchronizer is undecidable.
%

\subsection*{Notations}
If $i,j\in\N$, we denote $[i,j]$ the set $\{i,i+1,\dots,j\}$.
We will note $\B:=\{0,1\}$ the set of booleans.
If $X$ is a set, we denote $X^*:=\bigcup_{i\in\N} X^i$ the set of words on alphabet $X$. The empty word is denoted $\eps$. We will denote $u\pref v$ if $u$ is a prefix of $v$.
We will denote $\Sigma$ and $\Gamma$ for arbitrary finite alphabets throughout the paper.
If $u\in\Sigma^*$, we will denote $|u|$ its length and $\dom(u)=\{1,2,\dots,|u|\}$ its set of positions.

%

\section{Transductions}

\subsection{One-way Non-deterministic Transducers}

A \emph{one-way non-deterministic transducer} (1NT) is a tuple $T=\tuple{Q,\Sigma, \Gamma, \Delta, I,F}$, where
$Q$ is a finite set of states, $\Sigma$ is a finite input alphabet, $\Gamma$ is a finite output alphabet, $\Delta\subseteq Q\times(\Sigma\cup\{\eps\}) \times \Gamma^*\times Q$ is the transition relation, $I$ is the set of initial states, and $F$ the set of final states.
A transition $(p,a,v,q)$ of $\Delta$ will be denoted as $p\trans{a|v}q$.
A \emph{run} of $T$ on an input word $u\in\Sigma^*$ is a sequence of transitions $p_0\trans{a_1|v_1} p_1\trans{a_2|v_2}\dots\trans{a_n|v_n}p_n$, such that $u=a_1a_2\dots a_n$, $p_0\in I$ and $p_n\in F$.
The \emph{output} of this run is the word $v=v_1\dots v_n$.
The \emph{relation} computed by $T$ is $\sem{T}=\{(u,v)\mid \text{there exists a run of $T$ on $u$ with output $v$}\}\subseteq \Sigma^*\times\Gamma^*$.
To avoid unnecessary special cases, we will always assume throughout the paper that the input word $u$ is not empty.
Two transducers $T_1,T_2$ are \emph{classically equivalent} if $\sem{T_1}=\sem{T_2}$.
It is known from \cite{Griffiths68} that classical equivalence of 1NTs is undecidable.

\subsection{Two-way Transducers}

In 1NTs, transitions can either leave the reading head on the same input letter, or move it one step to the right. If the possibility of moving to the left is added, we obtain the model of \emph{two-way non-deterministic transducer} (2NT).
The transition relation is now of the form $\Delta\subseteq Q\times(\Sigma\cup\{\vdash,\dashv\}) \times \Gamma^*\times\{\toleft,\toright\}\times Q$, where the symbol $\vdash$ (resp. $\dashv$) marks the beginning (resp. end) of the input word. When reading this symbol, we forbid the production of a non-empty output, and the only allowed direction for transitions is $\toright$ (resp. $\toleft$). The semantics $\sem{T}\subseteq \Sigma^*\times\Gamma^*$ of a 2NT is defined in a natural way: the output of a run $p_0\trans{a_1|v_1,d_1} p_1\trans{a_2|v_2,d_2}\dots\trans{a_n|v_n, d_n}p_n$ is $v_1v_2\dots v_n$. See \cite{bose2018origin} for a formal definition.
Notice that $\eps$-transitions are not necessary anymore, since a transition $p\trans{\eps|v}q$ can be simulated by two transitions going right then left (or left then right if the symbol $\dashv$ is reached).

If the transition relation is deterministic, i.e. if for all $(p,a)\in Q\times(\Sigma\cup\{\vdash,\dashv\})$ there exists at most one $(v,d,q)\in \Gamma^*\times\{\toleft,\toright\}\times Q$ such that $p\trans{a|v,d}q$ is a transition in $\Delta$, we say that the transducer is a \emph{two-way deterministic transducer} (2DT).

Notice that the relation defined by a 2DT $T$ is necessarily a (partial) function: for all $u\in\Sigma^*$ there is at most one $v\in\Gamma^*$ such that $(u,v)\in\sem{T}$.
The class of functions definable by 2DTs is called \emph{regular string-to-string functions}. It has equivalent characterizations, such as MSO transductions \cite{engelfriet2001mso} and streaming transducers \cite{alur2010expressiveness}.
\subsection{Origin information}
The origin semantics was introduced in \cite{bojanczyk2014transducers} as an enrichment of the classical semantics for string-to-string transductions.
The principle is that the contribution of a run of $T$ to the semantics of $T$ is not only the input/output pair $(u,v)$, but an \emph{origin graph} describing how $v$ is produced from $u$ during this run.

Formally, an origin graph is a triple $(u,v,\orig)$ where $u\in\Sigma^*$, $v\in\Gamma^*$, and $\orig:\dom(v)\to\dom(u)$ associates to each position in $v$ a position in $u$: its \emph{origin}.   
An origin graph is associated to a run of a transducer $T$ in a natural way, by mapping to each position $y$ in $v$ the position $\orig(y)$ of the reading head in $u$ when writing to this position $y$. If an output is produced by an $\eps$-transition after the whole word has been processed in a 1NT, we take the last input letter as origin.
The \emph{origin semantics} $\osem{T}$ of $T$ is the set of origin graphs associated with runs of $T$.

\begin{example}\label{ex:idreversal}
The two following 2DTs $\Tid$ and $\Trev$ are classically equivalent and compute the identity relation $\{(a^n,a^n)\mid n\in\N\}$, but their origin semantics differ, as witnessed by their unique origin graphs for input $a^6$ given below.

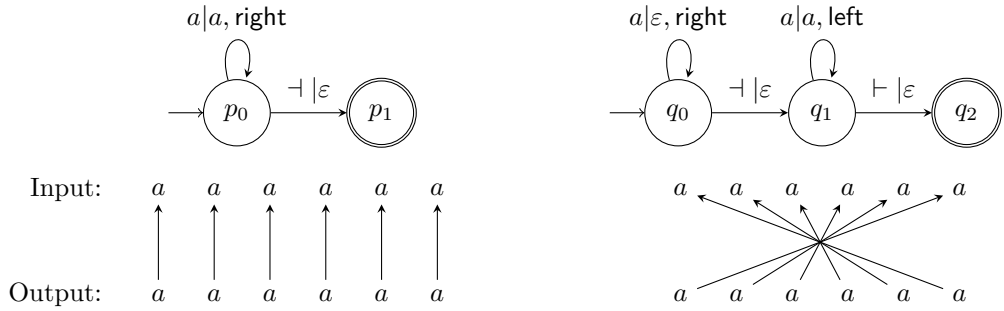
\begin{center}
\begin{tikzpicture}[scale=1,every node/.style={scale=1},node distance = .3cm, initial text=]
\node[initial,state] (init) {$p_0$};
\node[state,accepting, right=1cm of init] (suc) {$p_1$};

\path[->,>=stealth]
  (init) edge[loop above] node[above] {$a|a,\toright$} (init)
 (init) edge node[above] {$\dashv|\eps$} (suc)
  ;

 \node[initial,state, right=3cm of suc] (init2) {$q_0$};
\node[state,right =1cm of init2] (success2) {$q_1$};
\node[accepting,state,right =1cm of success2] (fin) {$q_2$};
\path[->,>=stealth]
  (init2) edge[loop above] node[above] {$a|\eps,\toright$} (init2)
  (init2) edge node[above] {$\dashv|\eps$} (success2)
  (success2) edge[loop above] node[above] {$a|a,\toleft$} (success2)
  (success2) edge node[above] {$\vdash|\eps$} (fin)
  ;

		\node[below left=.5cm and .5cm of init] (i1) {$a$};
		\node[right= of i1] (i2) {$a$};
		\node[right= of i2] (i3) {$a$};
		\node[right= of i3] (i4) {$a$};
		\node[right= of i4] (i5) {$a$};
		\node[right= of i5] (i6) {$a$};
		
		\node[below=1cm of i1] (o1) {$a$};
		\node[right= of o1] (o2) {$a$};
		\node[right= of o2] (o3) {$a$};
		\node[right= of o3] (o4) {$a$};
		\node[right= of o4] (o5) {$a$};
		\node[right= of o5] (o6) {$a$};

		\node[left=0.4cm of i1]	(labeli) {Input:};
		\node[left=0.4cm of o1] (labelo) {Output:};
		
		\draw [->,>=stealth] (o1) -- (i1);
		\draw [->,>=stealth] (o2) -- (i2);
		\draw [->,>=stealth] (o3) -- (i3);
		\draw [->,>=stealth] (o4) -- (i4);
		\draw [->,>=stealth] (o5) -- (i5);
		\draw [->,>=stealth] (o6) -- (i6);

\node[below right=.5cm and -.5cm of init2] (ai1) {$a$};
		\node[right= of ai1] (ai2) {$a$};
		\node[right= of ai2] (ai3) {$a$};
		\node[right= of ai3] (ai4) {$a$};
		\node[right= of ai4] (ai5) {$a$};
		\node[right= of ai5] (ai6) {$a$};
		
		\node[below=1cm of ai1] (ao1) {$a$};
		\node[right= of ao1] (ao2) {$a$};
		\node[right= of ao2] (ao3) {$a$};
		\node[right= of ao3] (ao4) {$a$};
		\node[right= of ao4] (ao5) {$a$};
		\node[right= of ao5] (ao6) {$a$};

		\draw [->,>=stealth] (ao1) -- (ai6);
		\draw [->,>=stealth] (ao2) -- (ai5);
		\draw [->,>=stealth] (ao3) -- (ai4);
		\draw [->,>=stealth] (ao4) -- (ai3);
		\draw [->,>=stealth] (ao5) -- (ai2);
		\draw [->,>=stealth] (ao6) -- (ai1);
		
\end{tikzpicture}
\end{center}

\end{example}

Two transducers are said \emph{origin equivalent} if they have the same origin semantics.
It is shown in \cite{bojanczyk2014transducers} that origin equivalence is decidable for regular transductions, and in \cite{bose2018origin} that origin equivalence is \pspace-complete for 2NTs.
See Appendix \ref{ap:ex_trans} for an example of two one-way transducers both computing the full relation $\Sigma^*\times\Gamma^*$, but not origin equivalent.

\section{MSO Resynchronizers}

While origin semantics gives a satisfying framework to recover decidability of transducer equivalence, it can be argued that this semantics is too rigid, as origin equivalence require that the output is produced in exactly the same way in both transducers.

In order to relax this constraint, the intermediate notion of \emph{resynchronization} has been introduced \cite{filiot2016equivalence,bose2018origin}. The idea is to let origins differ in a controlled way, while preserving the  input/output pair. Several notions of resynchronizations have been considered \cite{filiot2016equivalence,bose2018origin,bose2019synthesis}, we will focus in this work on MSO resynchronizers, also called regular resynchronizers.

\subsection{Regular languages and MSO}
We recall here how Monadic Second-Order logic (MSO) can be used to define languages. This framework will be then used to represent resynchronizers. 
Formulas of MSO are defined by the following grammar, where $a$ ranges over the alphabet $\Sigma$:
$$\varphi,\psi:=a(x)\mid x\leq y\mid x\in X\mid \exists x.\varphi\mid \exists X.\varphi\mid \varphi\vee\psi \mid \neg\varphi$$
Such formulas are evaluated on structures induced by finite words: the universe is the set of positions of the word, $a(x)$ means that position $x$ is labelled by letter $a$, and $x\leq y$ means that position $x$ occurs before position $y$.
Lowercase notation is used for first-order variables, ranging over positions of the word, and uppercase notation is used for second-order variables, ranging over sets of positions. 
Other classical operators such as $\wedge, \Rightarrow,\forall,=,+1,+2,\mathit{first},\mathit{last},\dots$ can be defined from this syntax and will be used freely. 
Let $\top$ be a tautology, defined for instance as $\exists x. a(x) \vee \neg (\exists x.a(x))$.

If $\varphi$ is an MSO formula and $u\in\Sigma^*$, we will note $u\models \varphi$ if $u$ is a model of $\varphi$, with classical MSO semantics.
The language $L(\varphi)$ defined by a closed formula $\varphi$ is $\{u\in\Sigma^*\mid u\models\varphi\}$.

If $\varphi$ contains free variables $X_1,\dots,X_n,x_1,\dots,x_k$, we can still define the language of $\varphi$, using an extended alphabet $\Sigma\times\B^{n+k}$. Extra boolean components at each  position are used to convey the values of free variables at this position: it is $1$ if the value of the second-order variable contains the position (resp. if the value of the first-order variable matches the position) and $0$ otherwise. The language of $\varphi$ is in this case a subset of $(\Sigma\times\B^{n+k})^*$, i.e. a set of words on $\Sigma$ enriched with valuations for the free variables.
If $I_1,\dots, I_n,i_1,\dots,i_k$  is an instantiation for the free variables of $\varphi$ in a word $u$, we will also write $(u,I_1,\dots, I_n,i_1,\dots,i_k)\models\varphi$ to signify that $u$ with this instantiation of the free variables satisfies $\varphi$.

For instance if $\varphi=\exists x.(x\in X\wedge a(x))$ uses a free second-order variable $X$, then the word $u=(a,0),(b,1),(a,1)\in (\Sigma\times\B)^*$ is a model of $\varphi$, denoted also $(aba,\{2,3\})\models\varphi$, but the word $(a,0),(b,1),(a,0)$ is not.
%
%

A language $L\subseteq(\Sigma\times\B^{n})^*$ is \emph{regular} if and only if there is a formula $\varphi$ of MSO with $n$ free variables recognizing $L$. This is equivalent to $L$ being recognizable by a deterministic finite automaton (DFA) on alphabet $\Sigma\times\B^n$ \cite{buchi60}.

\subsection{MSO Resynchronizers}
The principle behind MSO resynchronizers as defined in \cite{bose2018origin} is to describe in a regular way, with MSO formulas, how the origins can be redirected. This will induce a relation between sets of origin graphs: containment up to resynchronization.

The MSO formulas will be allowed to use a finite set of \emph{parameters}: extra information labelling the input word. This is reminiscent of the model of non-deterministic two-way transducers with \emph{common guess} \cite{bojanczyk2017classes}, where the guessing of extra parameters can be done in a consistent way through different visits of the same position in the input word.

\subsubsection{Definition}



We now define a subclass of regular resynchronizers from \cite{bose2018origin,bose2019synthesis}. We will see that for our purpose of resynchronizer synthesis, this subclass is equivalent to the full class of resynchronizers from \cite{bose2018origin,bose2019synthesis}. Intuitively, the full definition from \cite{bose2018origin,bose2019synthesis} allows to further restrict the semantics of a resynchronizer, which is not useful if we are just interested in the existence of a resynchronization between two transducers. This is further explained in Section \ref{sec:containment} and Appendix \ref{ap:alpha}.

Given an origin graph $\sigma=(u,v,\orig)$, an \emph{input parameter} is a subset of the input positions, encoded by a word on $\B$. Thus, a valuation for $m$ input parameters is given by a tuple $\bar{I} = (I_1, \dots, I_m)$ where for each $i\in[1,m], I_i \in\B^{|u|}$. 

The main differences between the following simplified definition and the one from \cite{bose2018origin,bose2019synthesis} is that we ignored output parameters (an extra labelling of the output word), and also removed extra formulas constraining the behaviour of the resynchronization with respect to both input and output parameters.

	\begin{definition}
		An MSO (or \emph{regular}) resynchronizer $R$ with $m$ input parameters is an MSO formula $\gamma$ with $m+2$ free variables $\gamma(\bar{I}, x,y)$, evaluated over the input word $u$.
	\end{definition}
	

%
%
%
Intuitively, $\gamma(\bar{I}, x,y)$ indicates that the origin $x$ of an output position can be redirected to a new origin $y$, as made precise in Definition \ref{def:semR}.
Although $R$ and $\gamma$ are actually the same object here, we will keep the two notations to maintain coherence with \cite{bose2018origin}, using $R$ for the abstract resynchronizer and $\gamma$ for the MSO formula, which is only one of the components of $R$ in \cite{bose2018origin}.
We now describe formally the semantics of an MSO resynchronizer.

	\begin{definition}\label{def:semR} \cite{bose2018origin}
		An MSO resynchronizer $R$ induces a relation $\sem{R}$ on origin graphs in the following way. If $\sigma = (u, v,\orig)$ and $\sigma' = (u',v',\orig')$ are two origin graphs, we have $(\sigma, \sigma')\in \sem{R}$ if and only if $u = u', v=v'$, and there exists input parameters $\bar{I} \in (\B^{|u|})^m$, such that for every output position $z\in dom(v)$, we have $(u, \bar{I}, \orig(z),\orig'(z)) \models \gamma$.
	\end{definition}
	

\subsubsection{Examples}\label{sec:ex_resynchr}
Plain blue arrows will represent the ``old'' origins in $\sigma$, and red dotted arrows the ``new'' origins in $\sigma'$. 

%
%
%
%
%
%
	
	\begin{example}\label{example:resync:univ}\cite{bose2018origin}
		The resynchronizer without parameters $\Runiv$, using only a tautology formula $\gamma=\top$, is called the universal resynchronizer, and resynchronizes any two origin graphs that share the same input and output. 
	\end{example}
	
	\begin{example}\label{example:oneshift}\cite{bose2018origin}
		The resynchronizer without parameters $\Rshift$ shifts all origins by exactly $1$ position left or right. This is achieved using a formula $\gamma(x,y) = (x=y+1)\vee (y=x+1)$.
	
\end{example}	

\begin{example}\label{ex:param}
The resynchronizer with one parameter defined by $\gamma= (I=\{x\}) \vee (x=y)$ allows at most one input position to be resynchronized to different origins.
\end{example}

\begin{center}
		\begin{tikzpicture}[scale=1, every node/.style={scale=1},node distance=.4cm]
		
		\node (i1) {$a$};
		\node[right= of i1] (i2) {$a$};
		\node[right= of i2] (i3) {$a$};
		\node[right= of i3] (i4) {$a$};
		\node[right= of i4] (i5) {$a$};
		\node[right= of i5] (i6) {$a$};
		
		\node[below=1cm of i1] (o1) {$b$};
		\node[below=1cm of i2] (o2) {$b$};
		\node[below=1cm of i3] (o3) {$b$};
		\node[below=1cm of i4] (o4) {$b$};
		\node[below=1cm of i5] (o5) {$b$};
		\node[below=1cm of i6] (o6) {$b$};

		\node[left=0.2cm of i1]	(labeli) {Input:};
		\node[left=0.2cm of o1] (labelo) {Output:};
		
		\draw [before] (o1) -- (i1);
		\draw [before] (o2) -- (i2);
		\draw [before] (o3) -- (i3);
		\draw [before] (o4) -- (i4);
		\draw [before] (o5) -- (i5);
		\draw [before] (o6) -- (i6);
		
		\draw [after] (o1) -- (i2);
		\draw [after] (o2) -- (i3);
		\draw [after] (o3) -- (i4);
		\draw [after] (o4) -- (i5);
		\draw [after] (o5) -- (i4);
		\draw [after] (o6) -- (i5);

		\node[right=2cm of i6] (Pi1) {$a$};
		\node[right= of Pi1] (Pi2) {$a$};
		\node[right= of Pi2] (Pi3) {$a$};
		\node[right= of Pi3] (Pi4) {$a$};
		\node[right= of Pi4] (Pi5) {$a$};
		\node[right= of Pi5] (Pi6) {$a$};
		
		\node[below=1cm of Pi1] (Po1) {$b$};
		\node[below=1cm of Pi2] (Po2) {$b$};
		\node[below=1cm of Pi3] (Po3) {$b$};
		\node[below=1cm of Pi4] (Po4) {$b$};
		\node[below=1cm of Pi5] (Po5) {$b$};
		\node[below=1cm of Pi6] (Po6) {$b$};

		\node[above right=.3cm and .5cm of i3, inner sep=0cm, anchor=center] {Example \ref{example:oneshift}};
		\node[above right=.3cm and .5cm of Pi3, inner sep=0cm,anchor=center] {Example \ref{ex:param}};
		
		
		\draw [before, transform canvas={xshift=-.2mm}] (Po1) -- (Pi1);
		\draw [before, transform canvas={xshift=-.2mm}] (Po2) -- (Pi2);
		\draw [before] (Po3) -- (Pi3);
		\draw [before] (Po4) -- (Pi3);
		\draw [before, transform canvas={xshift=-.2mm}] (Po5) -- (Pi5);
		\draw [before, transform canvas={xshift=-.2mm}] (Po6) -- (Pi6);
		
		\draw [after, transform canvas={xshift=.2mm}] (Po1) -- (Pi1);
		\draw [after, transform canvas={xshift=.2mm}] (Po2) -- (Pi2);
		\draw [after] (Po3) -- (Pi2);
		\draw [after] (Po4) -- (Pi6);
		\draw [after, transform canvas={xshift=.2mm}] (Po5) -- (Pi5);
		\draw [after, transform canvas={xshift=.2mm}] (Po6) -- (Pi6);
			
		\end{tikzpicture}

	\end{center}
	\subsection{Containment up to resynchronization}
	\begin{definition}\cite{bose2018origin}
		For a resynchronizer $R$ and two transducers $T_1,T_2$ we note $T_1 \subseteq R(T_2)$ if for every origin graph $\sigma_1 \in \osem{T_1}$, there exists $\sigma_2 \in \osem{T_2}$ such that $(\sigma_2, \sigma_1) \in \sem{R}$. 
	\end{definition}
	In other words this means that $\osem{T_1}$ is contained in the resynchronization expansion of $\osem{T_2}$.
	Examples can be found in Appendix \ref{ap:ex_resynchr}.

 For a fixed resynchronizer $R$ and a 2NT $T$, it might not be the case that $T \subseteq R(T)$, as witnessed by the resynchronizer $\Rshift$ from Example \ref{example:oneshift}. Moreover, if $T_1 \subseteq {R} (T_2)$ and $T_2 \subseteq {R}(T_3)$ it might not be the case that $T_1 \subseteq {R}(T_3)$, again this is examplified by $\Rshift$. This means that the containment relation up to a fixed resynchronizer $R$ is neither reflexive nor transitive in general.
	
	\subsection{Bounded resynchronizers}
	Note that the universal resynchronizer $\Runiv$ from Example \ref{example:resync:univ} relates any two graphs that share the same input and output. This causes the containment relation up to $\Runiv$ to boil down to classical containment, ignoring the origin information. I.e. we have $T_1 \subseteq {R_{univ}}(T_2)$ if and only if $\sem{T_1} \subseteq \sem{T_2}$.
This inclusion relation is undecidable, even in the case of one-way non-deterministic transducers \cite{Griffiths68}. Thus containment up to a fixed resynchronizer is undecidable in general, if no extra constraint is put on resynchronizers.
That is why the natural \emph{boundedness} restriction is introduced on MSO resynchronizers in \cite{bose2018origin}. 
	
	\begin{definition}\cite{bose2018origin} (Boundedness)
		A regular resynchronizer $R$ has bound $k$ if for all inputs $u$, input parameters $\bar{I}$, and target position $y\in dom(u)$, there are at most $k$ distinct positions $x_1, \dots x_k \in dom(u)$ such that $(u, \bar{I}, x_i, y) \models \gamma$ for all $i\in[1,k]$. A regular resynchronizer is bounded if it has bound $k$ for some $k\in \N$.
	\end{definition}

All examples of resynchronizations given in this paper (including Appendix) are bounded, except for $\Runiv$.
In Appendix \ref{ap:ex_resynchr}, we give examples of bounded resynchronizations that displace the origin by a distance that is not bounded.

Boundedness is a decidable property of MSO resynchronizers \cite[Prop. 15]{bose2018origin}. As stated in the next theorem, boundedness guarantees that the containment problem up to a fixed resynchronizer becomes decidable. Moreover, for any fixed bounded MSO resynchronizer, the complexity of this problem matches the complexity of containment with respect to strict origin semantics, or more simply the complexity of inclusion of non-deterministic automata.
	
\begin{theorem}\cite[Cor. 17]{bose2018origin}
For a fixed bounded MSO resynchronizer $R$ and given two 2NTs $T_1,T_2$, it is decidable in \pspace whether $T_1\subseteq R(T_2)$.
\end{theorem}

\section{Resynchronizability}

We will now be interested in the containment up to an unknown bounded resynchronizer.
Let us define the relation $\preceq$ on 2NTs by $T_1\preceq T_2$ if there exists a bounded resynchronizer $R$ such that $T_1\subseteq R(T_2)$.
This relation has been introduced in \cite{bose2019synthesis}, along with the same notion with respect to rational resynchronizers.

Focusing on bounded regular resynchronizers, the following result is obtained in \cite{bose2019synthesis}:
\begin{theorem}\cite{bose2019synthesis}
The relation $\preceq$ is decidable on unambiguous 2NTs.
\end{theorem}

The problem is left open in \cite{bose2019synthesis} for general 2NTs, and this is the purpose of the present work.
Now that the necessary notions have been presented, we move to our contributions.

\subsection{Containment relation}\label{sec:containment}

Let us start by expliciting a few properties of $\preceq$. First, let us emphasize that our simplified definition of MSO resynchronizer is justified by the fact that this definition yields the same relation $\preceq$ as the one from \cite{bose2018origin, bose2019synthesis}. This is fully explicited in Appendix \ref{ap:alpha}.

This simplified definition allows us to show basic properties of the $\preceq$ relation, see Appendix \ref{ap:transitive} for a detailed proof:
\begin{lemma}\label{lem:transitive}
The relation $\preceq$ is reflexive and transitive.
\end{lemma}
Since $\preceq$ is a pre-order, it induces an equivalence relation $\sim$ on 2NTs, defined by $\sim=\preceq\cap \succeq$.
Notice that this equivalence relation is intermediate between classical equivalence and origin equivalence, but it is not immediately clear that it does not coincide with classical equivalence. 

The following claim presents two pairs of transducers (one pair of 2DTs and one pair of 1NTs) equivalent for the classical semantics, but not $\sim$-equivalent.

\begin{claim}\label{claim:revnotpreceq}
\begin{itemize}
\item The 2NTs $\Tid$ and $\Trev$ from Example \ref{ex:idreversal} are not $\sim$-equivalent.
\item The two following 1NTs $\Ta,\Taa$ have the same classical semantics $\{(a^n,a^m)\mid n\leq m\leq 2n\}$, but are not $\sim$-equivalent.

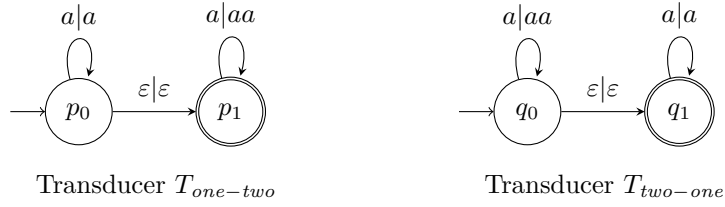
\begin{center}
\begin{tikzpicture}[node distance = 2cm, initial text=]
\node[initial,state] (init) {$p_0$};
\node[accepting,state,right of = init] (success) {$p_1$};
\node[below right=.4cm and .7cm of init,anchor=north] {Transducer $\Ta$};
\path[->,>=stealth]
  (init) edge[loop above] node[above] {$a|a$} (init)
  (init) edge node[above] {$\eps|\eps$} (success)
  (success) edge[loop above] node[above] {$a|aa$} (success)
  ;

 \node[initial,state, right=3cm of success] (init2) {$q_0$};
\node[accepting,state,right of = init2] (success2) {$q_1$};
\node[below right=.4cm and .7cm of init2,anchor=north] {Transducer $\Taa$};
\path[->,>=stealth]
  (init2) edge[loop above] node[above] {$a|aa$} (init2)
  (init2) edge node[above] {$\eps|\eps$} (success2)
  (success2) edge[loop above] node[above] {$a|a$} (success2)
  ;
\end{tikzpicture}
\end{center}
\end{itemize}
\end{claim}

A variant of the pair $\Ta,\Taa$ is presented in \cite[Example 5]{bose2019synthesis}, where it is claimed without proof that no bounded regular resynchronizer exists. 
A proof of Claim \ref{claim:revnotpreceq} will be obtained as a by-product of Theorem \ref{thm:transdesync} and explicited in Corollary \ref{cor:claim}.


\subsection{Limited traversal}
The goal of this section is to exhibit a pattern characterizing families of origin graphs that cannot be resynchronized with a bounded MSO resynchronizer.

	\begin{definition}\label{crossing}
		Let $\sigma=(u,v,\orig)$ and $\sigma'=(u,v,\orig')$ be two origin graphs with same input/output words. Given two input positions $x, z \in \dom(u)$, we say $x$ \emph{traverses} $z$ if there exists an output position $t\in \dom(v)$ with $orig(t)=x$ and either:
		\begin{itemize}
			\item $x \leq z $ and $orig'(t) > z$ (left to right traversal);
			\item $x \geq z$ and $orig'(t) < z$ (right to left traversal).  
		\end{itemize}
Intuitively, $x$ traverses $z$ if $x$ is resynchronized to some $y\neq z$, and $z$ is between the two positions $x,y$.

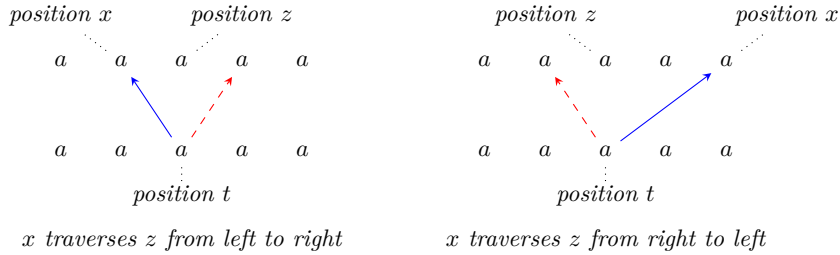
\begin{center}
	\begin{tikzpicture}[%
	scale=.9,
	every node/.style={
		text height=1ex,
		text depth=.25ex,
		node distance = .4cm,
		scale=.9,
	},
	]
		\node (i1) {$a$};
		\node[right= of i1] (i2) {$a$};
		\node[right= of i2] (i3) {$a$};
		\node[right= of i3] (i4) {$a$};
		\node[right= of i4] (i5) {$a$};
		
		\node[below=.8cm of i1] (o1) {$a$};
		\node[right= of o1] (o2) {$a$};
		\node[right= of o2] (o3) {$a$};
		\node[right= of o3] (o4) {$a$};
		\node[right= of o4] (o5) {$a$};
		
		\node[above right= 0.2cm and -0.2cm of i3] (labelx) {position $z$};
		\draw[dotted] (i3) -- (labelx);

		\node[above left= 0.2cm and -0.2cm of i2] (labelxprime) {position $x$};
		\draw[dotted] (i2) -- (labelxprime);

		\node (j1)[right = 2cm of i5] {$a$};
		\node[right= of j1] (j2) {$a$};
		\node[right= of j2] (j3) {$a$};
		\node[right= of j3] (j4) {$a$};
		\node[right= of j4] (j5) {$a$};
		
		\node[below=.8cm of j1] (m1) {$a$};
		\node[right= of m1] (m2) {$a$};
		\node[right= of m2] (m3) {$a$};
		\node[right= of m3] (m4) {$a$};
	
		\node[right= of m4] (m5) {$a$};
		
		\draw[before] (o3) -- (i2);
		\draw[after] (o3) -- (i4);
		
		\draw[before] (m3) -- (j5);
		\draw[after] (m3) -- (j2);
				
		\node[above left=0.2cm and -0.2cm of j3] (labelx1) {position $z$};
		\draw[dotted] (j3) -- (labelx1);
		
		\node[above right= 0.2cm and -0.2cm of j5] (labelxprime1) {position $x$};
		\draw[dotted] (j5) -- (labelxprime1);

		\node[below=.2cm of o3] (t1) {position $t$};
		\node[below=.2cm of m3] (t2) {position $t$};
		
		\draw[dotted] (o3) -- (t1);
		\draw[dotted] (m3) -- (t2);		
		
		\node[below=0.2cm of t1] {$x$ traverses $z$ from left to right};
		
		\node[below=0.2cm of t2] {$x$ traverses $z$ from right to left};

	\end{tikzpicture}
\end{center}

	\end{definition}
Let $k\in\N$, a pair of origin graphs $(\sigma,\sigma')$ on input/output words $(u,v)$ is said to have \emph{$k$-traversal} if for every $z\in\dom(u)$, there are at most $k$ distinct positions of $\dom(u)$ that traverse $z$.
A resynchronizer $R$ is said to have $k$-traversal if every pair of origin graphs $(\sigma, \sigma') \in \sem{R}$ has $k$-traversal. A resynchronizer $R$ has \emph{limited traversal} if there exists $k\in\N$ such that $R$ has $k$-traversal.

\newcommand{\Rparam}{\mathit{Right}}
\newcommand{\Lparam}{\mathit{Left}}
\newcommand{\Lsource}{\Lparam_i}
\newcommand{\Rsource}{\Rparam_i}
\newcommand{\Rtrav}{R_{\mathit{trav}}}
\newcommand{\Ltrav}{L_{\mathit{trav}}}

For any $k\in\N$ we want to construct a bounded resynchronizer $R_k$ that relates any pair of origin graphs that have $k$-traversal. We will use $2k$ input parameters: $\Rsource$ and $\Lsource$ for $i\in[0,k-1]$. Each parameter $\Rsource$ (resp. $\Lsource$) corresponds to a guessed set of input positions that may be redirected to the right (resp. left), but without traversing a position of the same set. For instance it is not possible for a position of $R_3$ to traverse another position of $R_3$ from left to right. Similarly, a position of $L_2$ cannot traverse another position of $L_2$ from right to left. We do not a priori require any of these sets to be disjoint from each other.  We construct $\gamma(x,y)=(x = y)\vee \Rtrav \vee\Ltrav$ to ensure this fact, where 
$$\Rtrav=\bigvee_{1\leq i \leq k} \big(x\in \Rsource \wedge x < y \wedge (\forall z\in [x+1,y]. z\not\in \Rsource)\big)$$
 verifies that positions labelled by the same $\Rsource$ do not traverse each other, and $\Ltrav$ does the same for the $\Lsource$ labels.
This achieves the description of the resynchronizer $R_k$, which will be proved correct in Lemmas \ref{lem:bounded} and \ref{lem:rk}.


\begin{lemma}\label{lem:bounded}
	The resynchronizer $R_k$ is bounded.
\end{lemma}
\begin{proof}
For each potential target position $y$, if two sources $x$ were labelled with the same input parameter, either one would traverse the other, or one would be at the left of $y$, which would contradict the definition of the formula. This means that if $\gamma(x,y)$ is valid then either $x=y$ or one of the parameters is used to indicate a single $x$ as source. There are only $2k$ parameters so for every input position $y$ there are at most $2k+1$ distinct positions $x$ such that $\gamma(x, y)$ is valid.
\end{proof}

\newcommand{\Rdist}{R_{\mathit{dist}}}
\newcommand{\free}{\mathit{FreeIndexes}}
\newcommand{\imin}{i_{\mathit{min}}}

	\begin{lemma}\label{lem:rk}
		If a pair of origin graphs $(\sigma,\sigma')$ has $k$-traversal, then $(\sigma, \sigma') \in \sem{R_k}$.
	\end{lemma}

	\ProofSketch

We describe an algorithm performing a left to right pass of the input word, and assigning labels $\Rparam_0,\Rparam_1,\dots,\Rparam_{k-1}$ to positions that are resynchronized to the right. We always assign to a position the minimal index currently available, in order to avoid the right traversal of any position by another position with the same label. We then show that under the hypothesis of $k$-traversal, this algorithm succeeds in finding an assignment of labels witnessing  $(\sigma, \sigma') \in \sem{R_k}$.
The same algorithm is then run in the other direction (right to left), to assign labels $\Lparam_i$.
See Appendix \ref{ap:proof_rk} for the full construction.
	\end{proof}
	
	\newcommand{\upref}{u_{\mathit{pref}}}
	\newcommand{\usuff}{u_{\mathit{suff}}}
	\newcommand{\Usuff}{U_{\mathit{suff}}}
	\begin{lemma}\label{lem:block}
		An MSO resynchronizer $R$ has limited traversal if and only if it is bounded.
	\end{lemma}

\begin{proof}

Let $m$ be the number of input parameters used in $R$.

($\Rightarrow$) Assume $R$ is not bounded, and let $k\in \N$, we want to build a pair $(\sigma,\sigma')\in\sem{R}$ exhibiting $k$-traversal. Since $R$ is not bounded, there exists a word $u\in\Sigma^*$, with input parameters $\bar{I}$, a position $y$, and a set $X$ of $2k+1$ distinct positions such that for all $x\in X$, we have $(u,\bar{I},x,y)\models\gamma$. Without loss of generality, we can assume that there are $k$ distinct positions $x_1,\dots x_k$ in $X$ that are strictly to the left of $y$. Let $a\in\Gamma$ be an arbitrary output letter and $v=a^k$. We define the origin graphs $\sigma,\sigma'$ on $(u,v)$ by setting for each $i\in [1,k]$ the origin of the $i^{th}$ letter of $v$ to $x_i$ in $\sigma$ and to $y$ in $\sigma'$. As witnessed by parameters $\bar{I}$, we have $(\sigma,\sigma')\in\sem{R}$. Moreover, the input position $y-1$ is traversed from left to right by $k$ different sources. Since $k$ is arbitrarily chosen, $R$ does not have limited traversal.

%
		
($\Leftarrow$) For the other direction, assume $R$ has no limited traversal. Let $\A$ be a deterministic automaton recognizing $\gamma$, on alphabet $\Sigma_\A=\Sigma\times \B^{m+2}$, and $Q$ be the state space of $\A$.
Let $k\in\N$ be arbitrary. There exists $(\sigma, \sigma') \in \sem{R}$ a pair of origin graphs on words $(u,v)$, and a position $z\in\dom(u)$ such that, without loss of generality, $z$ is traversed by $K=k \cdot |Q|$ positions $x_1<x_2<\dots< x_K$ from left to right, i.e. $x_K\leq z$. Let $\bar{I}$ be the input parameters witnessing $(\sigma,\sigma')\in\sem{R}$. This means that for each $i\in[1,K]$ there exists $y_i>z$ with $(u,\bar{I},x_i,y_i)\models\gamma$.
Let us split the input sequence $U=(u,\bar{I})\in\Sigma_\A^*$ according to position $z$: $U= wr$, where the last letter of $w$ is in position $z$. For each $i\in [1,K]$, let $w_i\in \Sigma_{\A}^*$ be the word $w$ with two extra boolean components: the source is marked by a bit $1$ in position $x_i$, and the target is left to be defined.
We know that for each $i$ there exists $r_i\in\Sigma_\A^*$ extending $r$ with a target position such that $w_i r_i$ is accepted by $\A$. Let $q_i$ be the state reached by $\A$ after reading $w_i$. By choice of $K$, there exists $q\in Q$ such that $q_i=q$ for $k$ distinct values $i_1,\dots i_k$ of $i$.
This means that for each $j\in[1,k]$, we have $w_{i_j}r_{i_1}$ accepted by $\A$, i.e. $(u,\bar{I},x_{i_j},y_{i_1})\models\gamma$. This achieves the proof that $R$ is not bounded.
\end{proof}

\begin{theorem}\label{thm:transdesync}
	Let $T_1, T_2$ be 2NTs. Then $T_1 \preceq T_2$ if and only if there exists $k\in\N$ such that for every $\sigma'\in \osem{T_1}$, there exists $\sigma\in\osem{T_2}$ with same input/output and $(\sigma, \sigma')$ has $k$-traversal.  
\end{theorem}
\begin{proof}
Assume such a bound $k$ exists. By Lemma \ref{lem:rk}, for every $\sigma'\in\osem{T_1}$ there exists $\sigma\in\osem{T_2}$ such that $(\sigma,\sigma')\in \sem{R_k}$. This implies $T_1 \subseteq R_k(T_2)$, and by Lemma \ref{lem:bounded} this $R_k$ is bounded thus witnessing $T_1\preceq T_2$.

Conversely, assume that no such bound $k$ exists, but that there is a bounded resynchronizer $R$ witnessing $T_1\preceq T_2$. By Lemma \ref{lem:block}, $R$ has $k$-traversal for some $k\in\N$. By assumption, there exists $\sigma'\in\osem{T_1}$ such that for all $\sigma\in\osem{T_2}$, $(\sigma,\sigma')$ does not have $k$-traversal. However, there must exists $\sigma$ such that $(\sigma,\sigma')\in\sem{R}$, contradicting the fact that $R$ has $k$-traversal.
\end{proof}

\begin{remark}
We have shown here that the resynchronizers $R_k$ are universal: if two transducers can be resynchronized, then this is witnessed by a resynchronizer $R_k$. This gives for instance a bound on the logical complexity of the MSO formulas needed in resynchronizers: the formula for $R_k$ is a disjunction of formulas using only one $\forall$ quantifier.
\end{remark}

Notice that unlike the existence of bounded resynchronizer, the notion of limited traversal is directly visible on pairs of origin graphs, and is therefore useful to prove that two transducers cannot be resynchronized. This is exemplified in the following corollary.

\begin{corollary}\label{cor:claim}
The transducers from Claim \ref{claim:revnotpreceq} are not $\sim$-equivalent. Indeed, in both cases, for a given input/output pair $(u,v)$ in the relation, only one pair $(\sigma,\sigma')$ of origin graphs is compatible with $(u,v)$, and these pairs of graphs exhibit traversal of arbitrary size.
\end{corollary}	
Here are visualizations of the phenomenon. The first picture shows a pair of graphs with $5$-traversal for $\Tid,\Trev$, witnessed by the only origin graphs on words $(a^{10},a^{10})$.
 The second picture does the same for the two 1NTs $\Ta,\Taa$, which has $3$-traversal on words $(a^{10},a^{15})$. In both cases, the input position being traversed is circled, and only origin arrows relevant to the traversal of this position are represented.

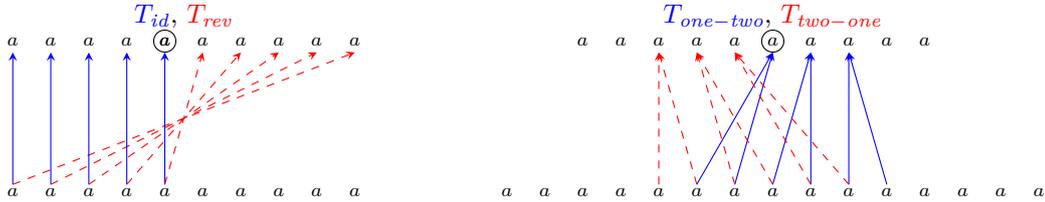
\begin{center}
	\begin{tikzpicture}[%
	every node/.style={
		font=\scriptsize,
		text height=1ex,
		text depth=.25ex,
	},
	]
	
	\def\h{-.15}

	\foreach \x in {0,.5,...,1.5}{
		\node at (\x,0) {$a$};
		\node at (\x,-2) {$a$};
	}
		\foreach \x in {2,2.5,...,4.5}{
		\node  at (\x,0) {$a$};
		\node at (\x,-2) {$a$};
	}
	
		\node[draw,circle, inner sep=.5pt]  at (2,0) {$a$};
	
	\node at (2.25,.3) {\normalsize \textcolor{blue}{$\Tid$}, \textcolor{red}{$\Trev$}};
	
	\foreach \x in {0,.5,...,2}{	
		\draw [before] (\x,-1.9) -- (\x,\h);
		\draw [after] (\x,-1.9) -- ({4.5-\x},\h);
	}
	
	
	\foreach \x in {7.5,8,...,9.5}{
	\node  at (\x,0) {$a$};
	}
	\foreach \x in {10.5,11,...,12}{
	\node  at (\x,0) {$a$};
	}
	
	\node[draw,circle, inner sep=.5pt] at (10,0) {$a$};

	\foreach \x in {6.5,7,...,13.5}{
		\node  at (\x,-2) {$a$};
	}

	\node at (10,.3) {\normalsize\textcolor{blue}{$\Ta$}, \textcolor{red}{$\Taa$}};
	

	\draw[before] (9,-1.9) -- (10,\h);
	\draw[before] (9.5,-1.9) -- (10,\h);
	\draw[before] (10,-1.9) -- (10.5,\h);
	\draw[before] (10.5,-1.9) -- (10.5,\h);
	\draw[before] (11,-1.9) -- (11,\h);
	\draw[before] (11.5,-1.9) -- (11,\h);
	
	\draw[after] (8.5,-1.9) -- (8.5,\h);	
	\draw[after] (9,-1.9) -- (8.5,\h);
	\draw[after] (9.5,-1.9) -- (9,\h);
	\draw[after] (10,-1.9) -- (9,\h);
	\draw[after] (10.5,-1.9) -- (9.5,\h);
	\draw[after] (11,-1.9) -- (9.5,\h);

	\end{tikzpicture}
\end{center}

\section{Undecidability of containment and equivalence}\label{sec:undec}
The aim of this section is to prove our main result:

\begin{theorem}\label{thm:undecidable}
		Given two 2NTs $T_1, T_2$, it is undecidable whether $T_1 \preceq T_2$. 
\end{theorem}
The result remains true if $T_1,T_2$ are 1NTs, with equivalence instead of containment, and if we restrict to any class of resynchronization that contains the ``shift resynchronizations'' : for each $k\in\N$, the $k$-shift resynchronization is defined by $\gamma(x,y)= (y\leq x\leq y+k)$.



We will proceed by reduction from the problem $\boundtape$, which asks given a deterministic Turing Machine $M$, whether it uses a bounded amount of its tape on empty input. For completeness, we prove in Appendix \ref{ap:boundtape} that this problem is undecidable, by a simple reduction from the Halting problem.
To perform the reduction from $\boundtape$ to the $\preceq$ relation, we first describe a classical construction used to encode runs of a Turing machine.

\subsection{The Domino Game}

Let $M$ be a deterministic Turing Machine with alphabet $A$, states $Q$, and transition table $\delta:Q\times A\to Q\times A\times\{\toleft,\toright\}$. Let $q_0$ (resp. $q_f$) be the initial (resp. final) state of $M$, and $B$ be the special blank symbol from the alphabet $A$, initially filling the tape.

Let $\#\notin A\cup Q$ be a new separation symbol, and $\Gamma=A\cup Q\cup\{\#\}$.

We sketch here a classical idea of using \emph{domino tiles} to simulate the run of a Turing Machine, for instance to prove undecidability of the Post Correspondence Problem \cite{PCP,Sipser}. See Appendix \ref{ap:domino} for the detailed construction of the set of tiles.

We encode successive \emph{configurations} of $M$ by words on $\Gamma^*$. The full run, or \emph{computation history} of $M$ is encoded by a finite or infinite word $\hist\in\Gamma^*\cup\Gamma^\omega$. We use a set of \emph{tiles} $D_M=\{(u_i,v_i)\in(\Gamma^*)^2\mid i\in\Sigma\}$, where $\Sigma$ is a finite alphabet of tile indexes. These tiles are designed to simulate the run of $M$ in the following sense (recall that $\pref$ stands for prefix):

\begin{lemma}\label{lem:domino}
Let $\lambda=i_1\dots i_k\in \Sigma^*$ be a sequence of tile indexes. Let $u_\lambda=u_{i_1}\dots u_{i_k}$, and $v_\lambda=q_0\#v_{i_1}\dots v_{i_k}$. If $\lambda$ is such that $u_\lambda\pref v_\lambda$, then we have $v_\lambda\pref \hist$.
\end{lemma}

We give here an example of how a run of $M$ is encoded, and how it is reflected on tiles:
\begin{example}\label{ex:tiles}
Consider the run of $M$ encoded by $q_0\#q_0B\#aq_1\#aq_1B\#q_2ab\#\in\Gamma^*$.
This is reflected by the following sequences of tiles:

\begin{center}
\begin{tikzpicture}[	font=\scriptsize,text height=1ex, text depth=.25ex,]
		
\def\i{.3}
\def\j{.4}
\def\k{.6}

\node at (-1,\k) {\normalsize $\lambda:$};
\node at (-1,0) {\normalsize $u_\lambda:$};
\node at (-1,-\k) {\normalsize $v_\lambda:$};

\draw (.1,-\k) node{$q_0\#$};

\draw (1,\k) node{$i_1$};
\draw (1,0) +(-\j,-\i) rectangle ++(\j,\i);
\draw (1,0) node{$q_0\#$};
\draw (1,-\k) +(-\j,-\i) rectangle ++(\j,\i);
\draw (1,-\k) node{$q_0B\#$};

\draw (2,\k) node{$i_2$};
\draw (2,0) +(-\j,-\i) rectangle ++(\j,\i);
\draw (2,0) node{$q_0B$};
\draw (2,-\k) +(-\j,-\i) rectangle ++(\j,\i);
\draw (2,-\k) node{$aq_1$};

\draw (2.9,\k) node{$i_3$};
\draw (2.9,0) +(-\i,-\i) rectangle ++(\i,\i);
\draw (2.9,0) node{$\#$};
\draw (2.9,-\k) +(-\i,-\i) rectangle ++(\i,\i);
\draw (2.9,-\k) node{$\#$};

\draw (3.7,\k) node{$i_4$};
\draw (3.7,0) +(-\i,-\i) rectangle ++(\i,\i);
\draw (3.7,0) node{$a$};
\draw (3.7,-\k) +(-\i,-\i) rectangle ++(\i,\i);
\draw (3.7,-\k) node{$a$};

\draw (4.6,\k) node{$i_5$};
\draw (4.6,0) +(-\j,-\i) rectangle ++(\j,\i);
\draw (4.6,0) node{$q_1\#$};
\draw (4.6,-\k) +(-\j,-\i) rectangle ++(\j,\i);
\draw (4.6,-\k) node{$q_1B\#$};

\draw (5.6,\k) node{$i_6$};
\draw (5.6,0) +(-\j,-\i) rectangle ++(\j,\i);
\draw (5.6,0) node{$aq_1B$};
\draw (5.6,-\k) +(-\j,-\i) rectangle ++(\j,\i);
\draw (5.6,-\k) node{$q_2ab$};

\draw (6.5,\k) node{$i_7$};
\draw (6.5,0) +(-\i,-\i) rectangle ++(\i,\i);
\draw (6.5,0) node{$\#$};
\draw (6.5,-\k) +(-\i,-\i) rectangle ++(\i,\i);
\draw (6.5,-\k) node{$\#$};
\end{tikzpicture}
\end{center}
\end{example}

\subsection{From tiles to transducers}

We now build two 1NTs $\Tup$ and $\Tdown$, based on the tiles of $D_M$.
The input alphabet of these transducers is the set $\Sigma$ of indexes of tiles of $D_M$. The output alphabet is $\Gamma$. Roughly, on input $i$, $\Tup$ outputs $u_i$ and $\Tdown$ outputs $v_i$. Additionally, $\Tup$ is allowed to non-deterministically start outputting a word that is not a prefix of $u_i$, and from there output anything in $\Gamma^*$. The transducer $\Tup$ is also allowed to output anything after the end of the input. The transducer $\Tdown$ starts by outputting $q_0\#$ at the beginning of the computation, so that on input $\lambda\in\Sigma^*$ it outputs $v_\lambda$.

The transducers $\Tup,\Tdown$ are pictured here, with $W_i=\{u\in\Gamma^*,|u|\leq |u_i|, u\not\pref u_i\}$:

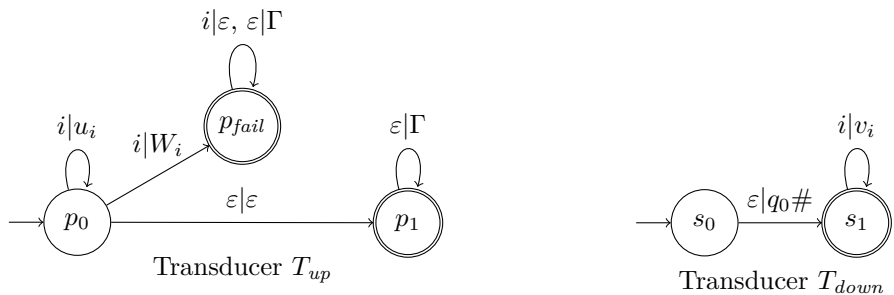
\begin{center}
\begin{tikzpicture}[scale=1,every node/.style={scale=1},node distance = 2cm, initial text=]
\node[initial,state] (init) {$p_0$};
\node[state, accepting, above right=.6cm and 1.5cm of init] (fail) {$\pfail$};
\node[accepting,state,below right= .6cm and 1.5cm of fail] (success) {$p_1$};

\node[below=1.1cm of fail,anchor=north] {Transducer $\Tup$};

\path[->]
(init) edge[loop above] node[above] {$i|u_i$} (init)
  (init) edge node[above=.1cm] {$i|W_i$} (fail)
  (fail) edge[loop above] node[above] {$i|\eps$, $\eps|\Gamma$} (fail)
  
  (init) edge node[above] {$\eps|\eps$} (success)
  (success) edge[loop above] node[above] {$\eps|\Gamma$} (success)
  ;

 \node[initial,state, right=3cm of success] (init2) {$s_0$};
\node[accepting,state,right of = init2] (success2) {$s_1$};
\node[below right=.2cm and .7cm of init2,anchor=north] {Transducer $\Tdown$};
\path[->]
  (init2) edge node[above] {$\eps|q_0\#$}  (success2)
  (success2) edge[loop above] node[above] {$i|v_i$} (success2)
  ;
\end{tikzpicture}
\end{center}

The main idea of this construction is that if $\lambda=i_1\dots i_k\in\Sigma^*$ is such that $u_\lambda\pref v_\lambda$ follow $\hist$ as in Example \ref{ex:tiles}, then on input $\lambda$, $\Tdown$ outputs $v_\lambda$, the only matching computation of $\Tup$ starts by outputting $u_\lambda$, and the bound on traversal will (roughly) match the size of the tape used by $M$ in this prefix of the computation. Indeed, if $\Tup$ and $\Tdown$ output the encoding of the same configuration of size $K$ on disjoint inputs, it witnesses a traversal of size roughly $K$ (``roughly'' because tiles allow up to three output letters on one input letter). The extra part of $\Tup$ is used to guarantee that $\sem{\Tdown}\subseteq\sem{\Tup}$ holds, even in cases when the input $\lambda$ does not correspond to a prefix of the computation of $M$.

\begin{example}\label{ex:updown}
Let $\lambda=i_1i_2\dots i_7$ be the sequence of tile indexes from Example \ref{ex:tiles}.
We show here a $2$-traversal exhibited by $\Tup,\Tdown$ on input $\lambda$. The traversed input position is circled, and only arrows relevant to the traversal of this position are represented.

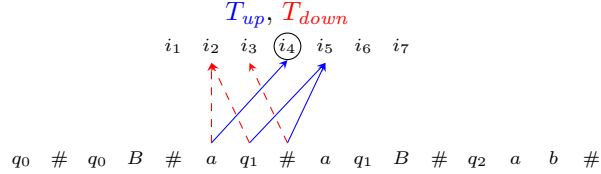
\begin{center}
	\begin{tikzpicture}[%
	every node/.style={
		font=\scriptsize,
		text height=1ex,
		text depth=.25ex,
	},
	]
	\foreach \i in {1,...,3} {
		\node (i\i) at ({\i/2},0) {$i_\i$};

	}
	\foreach \i in {5,...,7} {
		\node (i\i) at ({\i/2},0) {$i_\i$};

	}
	
	\node[draw,circle, inner sep=.5pt] (i4) at ({4/2},0) {$i_4$};
	
	\def\offset{-2}
	\foreach[count=\i] \gam in {q_0,\#,q_0,B,\#,a,q_1,\#,a,q_1,B,\#,q_2,a,b,\#} {
		\node (o\i) at ({\offset+\i/2},-1.5) {$\gam$};
	}	
	
	\node at (2,.4) {\normalsize \textcolor{blue}{$\Tup$}, \textcolor{red}{$\Tdown$}};

		\draw [before] (o6.north) -- (i4.south);
		\draw [before] (o7.north) -- (i5.south);
		\draw [before] (o8.north) -- (i5.south);
		\draw [after] (o6.north) -- (i2.south);
		\draw [after] (o7.north) -- (i2.south);
		\draw [after] (o8.north) -- (i3.south);

	\end{tikzpicture}
\end{center}		
\end{example}

\begin{theorem}\label{thm:reduction}
We have $\Tdown\preceq\Tup$ if and only if $M\in\boundtape$.
\end{theorem}

\begin{proof}
First, assume $M\in\boundtape$, let $K$ be the bound on the tape size used by $M$. Let $R$ be the resynchronization that shifts by at most $K+2$ positions to the left, via $\gamma(x,y)=(y\leq x) \wedge (x\leq y+K+2)$. We claim that $\Tdown\subseteq R(\Tup)$. It is clear that $R$ is bounded.
Let $\sigma'\in\osem{\Tdown}$ be an origin graph $(\lambda,v,\orig')$. Notice that by definition of $\Tdown$, we have $v=v_\lambda= q_0\#v_{i_1}\dots v_{i_n}$ on input $\lambda = i_1 \dots i_n$. 
We now distinguish two cases:
\begin{itemize}
\item If $u_\lambda\pref v_\lambda$, then by Lemma \ref{lem:domino}, we have $v_\lambda\pref\hist$. The transducer $\Tup$ is able to output $v_\lambda$ without going through the state $\pfail$, with a shift of one configuration as seen in Example \ref{ex:updown}. It only needs to pad $u_\lambda$ with the last configuration in state $p_1$.  Let $\sigma$ be the origin graph for this run. Since the encoding of a configuration has size at most $K+2$, we have $(\sigma,\sigma')\in\sem{R}$.
\item If $u_\lambda\not\pref v_\lambda$, let $\lambda' \pref \lambda$ be the longest prefix such that $u_{\lambda'} \pref v_{\lambda}$ . Now in order to output $v_\lambda$, the transducer $\Tup$ has to output $u_{\lambda'}$ in $p_0$ when processing $\lambda'$. After processing $\lambda'$, the transducer $\Tup$ is forced to move to state $\pfail$ in order to match the output of $\Tdown$. From this state $\Tup$ is allowed to output anything from any positions, so in particular there exists a run where the remaining output of $v_{\lambda'}$ is produced immediately, then $\Tup$ synchronizes with $\Tdown$ during the next configuration encoding, and finally the rest of the desired output $v_\lambda$ is produced on the same input positions as in $\Tdown$. As before, the shift when processing $\lambda$ is at most $K+2$, and therefore this run induces an origin graph $\sigma$ with $(\sigma,\sigma')\in\sem{R}$.
\end{itemize}

We now assume $M\notin\boundtape$. We want to use Theorem \ref{thm:transdesync} to conclude that $\Tdown\not\preceq\Tup$. Let $k\in\N$, and $\lambda\in\Sigma^*$ such that $u_\lambda\pref v_\lambda$ and $u_\lambda$ is a prefix of $\hist$ witnessing a configuration of size $k+2$. Let $\sigma'$ be the only origin graph of $\Tdown$ on input $\lambda$, with output $v_\lambda$.
There is only one way for $\Tup$ to output $v_\lambda$ on input $\lambda$: it is by using a run avoiding $\pfail$. Let $\sigma\in\osem{\Tup}$ be the corresponding origin graph.
Since $\Tup$ is one configuration behind, and since a configuration of size $k+2$ is produced by at least $k$ inputs, the pair $(\sigma,\sigma')$ has a position traversed $k$ times. This is true for arbitrary $k$, so by Theorem \ref{thm:transdesync}, we can conclude that $\Tdown\not\preceq\Tup$.
\end{proof}

Since $\boundtape$ is undecidable, this achieves the proof of Theorem \ref{thm:undecidable}.

Notice that in the case where $M\in\boundtape$, the resynchronization does not need parameters, and can be restricted to some simple classes of resynchronizations. This is stated in the following corollary:

\begin{corollary}\label{cor:noparam}
Given $T_1,T_2$ two 1NTs, it is undecidable whether $T_1\preceq T_2$. This result still holds when considering any restricted class of resynchronizers that contains the $k$-shift resynchronizers.
\end{corollary}

We can also strengthen the above proof to show undecidability of equivalence up to some unknown resynchronization:
\begin{theorem}
Given $T_1,T_2$ two 1NTs, it is undecidable whether $T_1\sim T_2$.
\end{theorem}
\begin{proof}
It suffices to take $\Tdown'=\Tdown\cup\Tup$ in the above proof. This way we clearly have $\Tup\preceq\Tdown'$, and the other direction $\Tdown'\preceq\Tup$ is equivalent to $\Tdown\preceq\Tup$, so it reduces to $\boundtape$ as well.
\end{proof}
%
Finally, let us mention that this proof allows us to recover and strengthen undecidability results on \emph{rational} transducers from \cite{bose2019synthesis}. We recall the definition of rational transducers in Appendix \ref{ap:rational}.

Since the shift resynchronizations are rational, and that any rational resynchronization is in particular bounded regular \cite[Theorem 3]{bose2019synthesis}, our reduction can be used in particular as an alternative proof of undecidability of rational resynchronization synthesis, shown in \cite{bose2019synthesis} via one-counter automata. 
This means we directly obtain this corollary:
\begin{restatable}{corollary}{rational}\label{cor:rational}
Given two 1NTs $T_1,T_2$ such that $\sem{T_1}\subseteq \sem{T_2}$, it is undecidable whether there exists a rational resynchronizer $\Rrat$ such that $T_1 \subseteq \Rrat(T_2)$.	
\end{restatable}

We can further strengthen the result via the following theorem:


\begin{restatable}{theorem}{ratreg}\label{thm:rationalundecidable}
	Given two 1NTs $T_1,T_2$ and a regular resynchronizer $\Rreg$ such that $T_1 \subseteq \Rreg(T_2)$, it is undecidable whether there exists a rational resynchronizer $\Rrat$ such that $T_1 \subseteq \Rrat(T_2)$.	
\end{restatable}

Due to space constraints, the proof is presented in Appendix \ref{ap:rational}.



\section{Conclusion}
In this work we investigated the containment relation on transducers up to unknown regular resynchronization. We showed that this relation forms a pre-order, strictly between classical containment and containment with respect to origin semantics. We introduced a syntactical condition called limited traversal, characterizing resynchronizable transducers pairs. Using this tool we proved that the resynchronizer synthesis is undecidable already in the case of 1NTs, while the problem was left open for 2NTs in \cite{bose2019synthesis}.

We leave open the decidability of the resynchronizability relation on \emph{functional} transducers. Since our construction highly uses non-functionality, it seems a different approach is needed. 

\newpage
\bibliography{bibliography}
\newpage
\appendix
\section{Appendix}

\subsection{Examples of transducers}\label{ap:ex_trans}

\begin{example}\label{ex:fulltrans}
Two equivalent transducers computing the full relation $\Sigma^*\times\Gamma^*$. Notice that $\eps$-transitions are necessary to compute this relation.
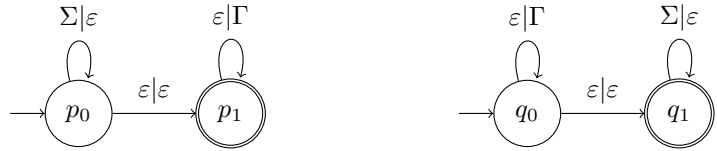
\begin{center}
\begin{tikzpicture}[node distance = 2cm, initial text=]
\node[initial,state] (init) {$p_0$};
\node[accepting,state,right of = init] (success) {$p_1$};

\path[->]
  (init) edge[loop above] node[above] {$\Sigma|\eps$} (init)
  (init) edge node[above] {$\eps|\eps$} (success)
  (success) edge[loop above] node[above] {$\eps|\Gamma$} (success)
  ;

 \node[initial,state, right=3cm of success] (init2) {$q_0$};
\node[accepting,state,right of = init2] (success2) {$q_1$};

\path[->]
  (init2) edge[loop above] node[above] {$\eps|\Gamma$} (init2)
  (init2) edge node[above] {$\eps|\eps$} (success2)
  (success2) edge[loop above] node[above] {$\Sigma|\eps$} (success2)
  ;
\end{tikzpicture}
\end{center}
	
\end{example}

\begin{example}\label{ex:fullorigin}
Consider the two transducers from Example \ref{ex:fulltrans} with $\Sigma=\{a,b\}$ and $\Gamma=\{c,d\}$.
Although they are equivalent in the classical sense as they compute the full relation $\Sigma^*\times \Gamma^*$, their origin semantics is different, as witnessed by the following examples of origin graphs on input $u=abbaba$ and output $v=cdddcc$. 

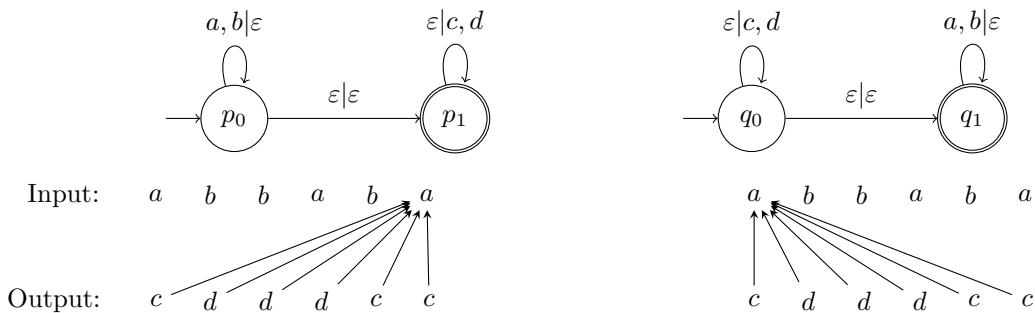
\begin{center}
\begin{tikzpicture}[node distance = .3cm, initial text=]
\node[initial,state] (init) {$p_0$};
\node[accepting,state,right =2cm of init] (success) {$p_1$};

\path[->]
  (init) edge[loop above] node[above] {$a,b|\eps$} (init)
  (init) edge node[above] {$\eps|\eps$} (success)
  (success) edge[loop above] node[above] {$\eps|c,d$} (success)
  ;

 \node[initial,state, right=3cm of success] (init2) {$q_0$};
\node[accepting,state,right=2cm of init2] (success2) {$q_1$};

\path[->]
  (init2) edge[loop above] node[above] {$\eps|c,d$} (init2)
  (init2) edge node[above] {$\eps|\eps$} (success2)
  (success2) edge[loop above] node[above] {$a,b|\eps$} (success2)
  ;
		
		\node[below left=.5cm and .5cm of init] (i1) {$a$};
		\node[right= of i1] (i2) {$b$};
		\node[right= of i2] (i3) {$b$};
		\node[right= of i3] (i4) {$a$};
		\node[right= of i4] (i5) {$b$};
		\node[right= of i5] (i6) {$a$};
		
		\node[below=1cm of i1] (o1) {$c$};
		\node[right= of o1] (o2) {$d$};
		\node[right= of o2] (o3) {$d$};
		\node[right= of o3] (o4) {$d$};
		\node[right= of o4] (o5) {$c$};
		\node[right= of o5] (o6) {$c$};

		\node[left=0.4cm of i1]	(labeli) {Input:};
		\node[left=0.4cm of o1] (labelo) {Output:};
		
		\draw [->,>=stealth] (o1) -- (i6);
		\draw [->,>=stealth] (o2) -- (i6);
		\draw [->,>=stealth] (o3) -- (i6);
		\draw [->,>=stealth] (o4) -- (i6);
		\draw [->,>=stealth] (o5) -- (i6);
		\draw [->,>=stealth] (o6) -- (i6);

\node[below right=.5cm and -.5cm of init2] (ai1) {$a$};
		\node[right= of ai1] (ai2) {$b$};
		\node[right= of ai2] (ai3) {$b$};
		\node[right= of ai3] (ai4) {$a$};
		\node[right= of ai4] (ai5) {$b$};
		\node[right= of ai5] (ai6) {$a$};
		
		\node[below=1cm of ai1] (ao1) {$c$};
		\node[right= of ao1] (ao2) {$d$};
		\node[right= of ao2] (ao3) {$d$};
		\node[right= of ao3] (ao4) {$d$};
		\node[right= of ao4] (ao5) {$c$};
		\node[right= of ao5] (ao6) {$c$};

		\draw [->,>=stealth] (ao1) -- (ai1);
		\draw [->,>=stealth] (ao2) -- (ai1);
		\draw [->,>=stealth] (ao3) -- (ai1);
		\draw [->,>=stealth] (ao4) -- (ai1);
		\draw [->,>=stealth] (ao5) -- (ai1);
		\draw [->,>=stealth] (ao6) -- (ai1);
		
\end{tikzpicture}
\end{center}

\end{example}

\subsection{Examples of resynchronizers}\label{ap:ex_resynchr}

\newcommand{\Trightleft}{T_{\rightarrow\leftarrow}}

	\begin{example}\label{example:resynch:long}
		The resynchronizer without parameters $\Rblock$ behaves as follows: if the origin is the first letter of an $a$-block, then it is moved to the last letter of this $a$-block. If the origin is a $b$ then it does not change.
  
$$\begin{array}{ll}
\gamma(x,y) =& (x\leq y \wedge (\forall z \in [x,y]. a(z)) \wedge \neg a(x-1) \wedge \neg a(y+1))\\
&\bigvee~(b(x)\wedge x=y)
\end{array}
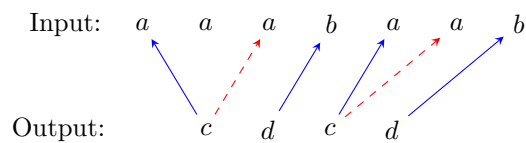
$$

	\begin{center}
		\begin{tikzpicture}[node distance=.4cm]
		
		\node (i1) {$a$};
		\node[right= of i1] (i2) {$a$};
		\node[right= of i2] (i3) {$a$};
		\node[right= of i3] (i4) {$b$};
		\node[right= of i4] (i5) {$a$};
		\node[right= of i5] (i6) {$a$};
		\node[right= of i6] (i7) {$b$};
		
		\node[below=1cm of i2] (o1) {$c$};
		\node[right= of o1] (o2) {$d$};
		\node[right= of o2] (o3) {$c$};
		\node[right= of o3] (o4) {$d$};

		\node[left=0.2cm of i1]	(labeli) {Input:};
		\node[left=1.0cm of o1] (labelo) {Output:};
		
		\draw [before] (o1) -- (i1);
		\draw [before] (o2) -- (i4);
		\draw [before] (o3) -- (i5);
		\draw [before] (o4) -- (i7);
		
		\draw [after] (o1) -- (i3);
		\draw [after] (o3) -- (i6);
		\end{tikzpicture}
	\end{center}
	Here is an example of behaviour of the same resynchronizer, applied to a two-way transducer  $\Trightleft$ doing two passes of the input word, one left-to-right and one-right-to-left, and outputting a new letter at each alternation of input letters $a$ and $b$.
		\begin{center}
		\begin{tikzpicture}[node distance=.4cm]
		
		\node (i1) {$a$};
		\node[right= of i1] (i2) {$a$};
		\node[right= of i2] (i3) {$a$};
		\node[right= of i3] (i4) {$b$};
		\node[right= of i4] (i5) {$a$};
		\node[right= of i5] (i6) {$a$};
		\node[right= of i6] (i7) {$b$};
		
		\node[below=1cm of i1] (o1) {$c$};
		\node[right= of o1] (o2) {$d$};
		\node[right= of o2] (o3) {$c$};
		\node[right= of o3] (o4) {$d$};
		\node[right= of o4] (o5) {$c$};
		\node[right= of o5] (o6) {$d$};
		\node[right= of o6] (o7) {$c$};

		\node[left=0.2cm of i1]	(labeli) {Input:};
		\node[left=0.2cm of o1] (labelo) {Output:};
		
		\draw [before] (o1) -- (i1);
		\draw [before] (o2) -- (i4);
		\draw [before] (o3) -- (i5);
		\draw [before] (o4) -- (i7);
		\draw [before] (o5) -- (i6);
		\draw [before] (o6) -- (i4);
		\draw [before] (o7) -- (i3);
		
		\draw [after] (o1) -- (i3);
		\draw [after] (o3) -- (i6);
		\end{tikzpicture}
	\end{center}	
\end{example}
\begin{example}\label{ex:1tolast}\cite{bose2018origin}
We give the example of $\Rfirstlast=(\top,\top,\gamma,\top)$: a resynchronizer without parameters, with $\gamma(x,y)=(x=\textit{first})\wedge(y=\textit{last})$, allowing only the resynchronization of origins from the first input position to the last one, and no other origins in the new origin graph.
	\begin{center}
	
	\begin{tikzpicture}[node distance=.4cm, inner sep=.1cm]
		\node(i1) {$a$};
		\node[right= of i1] (i2) {$b$};
		\node[right= of i2] (i3) {$b$};
		\node[right= of i3] (i4) {$a$};
		\node[right= of i4] (i5) {$b$};
		\node[right= of i5] (i6) {$a$};
		
		\node[below=1cm of i1] (o1) {$c$};
		\node[right= of o1] (o2) {$d$};
		\node[right= of o2] (o3) {$d$};
		\node[right= of o3] (o4) {$d$};
		\node[right= of o4] (o5) {$c$};
		\node[right= of o5] (o6) {$c$};

		\node[left=0.4cm of i1]	(labeli) {Input:};
		\node[left=0.4cm of o1] (labelo) {Output:};
		
		\draw [after] (o1) -- (i6);
		\draw [after] (o2) -- (i6);
		\draw [after] (o3) -- (i6);
		\draw [after] (o4) -- (i6);
		\draw [after] (o5) -- (i6);
		\draw [after] (o6) -- (i6);
		
		\draw [before] (o1) -- (i1);
		\draw [before] (o2) -- (i1);
		\draw [before] (o3) -- (i1);
		\draw [before] (o4) -- (i1);
		\draw [before] (o5) -- (i1);
		\draw [before] (o6) -- (i1);
	\end{tikzpicture}
\end{center}
\end{example}

\begin{example}\label{ex:contain1stlast}
Let $\Tlast,\Tfirst$ be the two transducers from Example \ref{ex:fullorigin}, and $\Rfirstlast$ the MSO resynchronizer from Example \ref{ex:1tolast}.
Then we have $\Tfirst\subseteq \Rfirstlast(\Tlast)$.
\end{example}

\begin{example}\label{ex:fastslow}
Let us give an example of two transducers $\Tfast$,$\Tslow$ with $\sem{\Tfast}=\sem{\Tslow}=\{(a^n,a^m)\mid n,m\in\N\}$, and $\Tslow\preceq\Tfast$ but $\Tfast\not\preceq\Tslow$.
\begin{center}
\begin{tikzpicture}[node distance = 2cm, initial text=]
\node[initial,state] (init) {$p_0$};
\node[accepting,state,right of = init] (success) {$p_1$};

\node[below right=.4cm and.7cm of init,anchor=north] {Transducer $\Tfast$};

\path[->]
  (init) edge[loop above] node[above] {$\eps|a$} (init)
  (init) edge node[above] {$\eps|\eps$} (success)
  (success) edge[loop above] node[above] {$a|\eps$} (success)
  ;

 \node[initial,state, right=3cm of success] (init2) {$q_0$};
\node[accepting,state,above right =.6cm and 2cm of init2] (success2) {$q_1$};
\node[accepting,state,below right =.6cm and 2cm of init2] (succ3) {$q_2$};

\node[below right=.4cm and .1cm of init2,anchor=north] {Transducer $\Tslow$};

\path[->]
  (init2) edge[loop above] node[above] {$a|a$} (init2)
  (init2) edge node[above] {$a|\eps$} (success2)
  (success2) edge[loop above] node[above] {$a|\eps$} (success2)
    (init2) edge node[above] {$\eps|a$} (succ3)
  (succ3) edge[loop above] node[above] {$\eps|a$} (succ3)
  ;
\end{tikzpicture}
\end{center}
Indeed, we have $\Tslow\subseteq R(\Tfast)$ where $R$ uses only $\gamma(x,y)=(x=\textit{first})$, which is bounded. However, if we had $\Tfast\subseteq R'(\Tslow)$, then $R'$ would need to redirect arbitrarily many positions to the first one, and therefore it could not be bounded.
\end{example}

\subsection{The original definition of resynchronizers}\label{ap:alpha}

We now give the original definition of MSO resynchronizers from \cite{bose2018origin,bose2019synthesis}, that we will call here \emph{extended MSO resynchronizer}, to emphasize the difference with our simplified version..

In addition to input parameters, extended MSO resynchronizers are also allowed to guess \emph{output parameters}, labelling the output word.

Given an origin graph $\sigma=(u,v,\orig)$, an output parameter is a subset of the output positions, encoded by a word on $\B$. Thus, a valuation for $n$ output parameters are given by $\bar{O} = (O_1, \dots, O_n) \in (\B^{|v|})^n$. 
Given an output alphabet $\Gamma$ and a number $n$ of output parameters, we define the set of \emph{output-types} as $\Gamma\times\B^n$. The role of an output-type is to describe a possible labelling of an output position, including the value of output parameters.
More precisely, given $v\in \Gamma^*$, $\bar{O}=(O_1,\dots,O_m)\in(\B^{|v|})^n$ and $x\in\dom(v)$, we call output-type of $x$ the element $\tau=(a,b_1,\dots,b_m)\in\Gamma\times\B^n$ obtained by projecting each coordinate of $(v,O_1,\dots,O_m)$ onto its $x^{th}$ position.
Notice that in the absence of output parameters, an output-type is simply a letter from $\Gamma$.

We can now give the definition of extended MSO resynchronizers:

	\begin{definition}\cite{bose2018origin}
		An MSO resynchronizer $R$ with $m$ input parameters and $n$ output parameters is a tuple $(\alpha, \beta, \gamma, \delta)$, where
		\begin{itemize}
			\item $\alpha(\bar{I})$ is an MSO formula over the input word with input parameters $\bar{I}=(I_1,\dots,I_m)$.
			\item $\beta(\bar{O})$ is an MSO formula over the output word with output parameters $\bar{O}=(O_1,\dots,O_n)$.  
			\item For every output-type $\tau\in\Gamma \times \B^n$, $\gamma(\tau)$ is an MSO formula with $m+2$ free variables: $\gamma(\tau)(\bar{I}, x,y)$ over the input word $u$, that indicates that the origin $x$ of an output position of type $\tau$ can be redirected to a new origin $y$.
			\item For every pair of output-types $\tau_1,\tau_2$, $\delta(\tau_1,\tau_2)$ is an MSO formula with $m+2$ free variables: $\delta(\tau_1,\tau_2)(\bar{I},z_1,z_2)$ over the input word $u$ is required to hold if $z_1,z_2$ are the new origins of two consecutive output positions $x_1,x_2$ with type $\tau_1,\tau_2$ respectively.
		\end{itemize}
	\end{definition}


We now describe formally the semantics of a extended MSO resynchronizer.

	\begin{definition} \cite{bose2018origin}
		An MSO resynchronizer $R=(\alpha,\beta,\gamma,\delta)$ induces a relation $\sem{R}$ on origin graphs in the following way. If $\sigma = (u, v,\orig)$ and $\sigma' = (u',v',\orig')$ are two origin graphs, we have $(\sigma, \sigma')\in \sem{R}$ if and only if $u = u', v=v'$, and there exists input parameters $\bar{I} \in (\B^{|u|})^m$, $\bar{O} \in (\B^{|v|})^n$, such that the following requirements hold:
		\begin{itemize}
			\item $(u, \bar{I}) \models \alpha$
			\item $(v, \bar{O}) \models \beta$
			\item For every output position $x\in dom(v)$ of type $\tau$, we have $(u, \bar{I}, \orig(x),\orig'(x)) \models \gamma(\tau)$
			\item For all consecutive output positions $x_1,x_2\in \dom(v)$ of type $\tau_1,\tau_2$ respectively, we have $(u,\bar{I},\orig'(x_1),\orig'(x_2)) \models \delta(\tau_1,\tau_2)$. 
		\end{itemize}
	\end{definition}

For examples making use of all components, see \cite{bose2018origin}.

We also recall the definition of boundedness for extended MSO resynchronizers:
	\begin{definition}\cite{bose2018origin} (Boundedness)
		A regular resynchronizer $R$ has bound $k$ if for all inputs $u$, input parameters $\bar{I}$, output-types $\tau\in \Gamma \times \B^n$, and target position $y\in dom(u)$, there are at most $k$ distinct positions $x_1, \dots x_k \in dom(u)$ such that $(u, \bar{I}, x_i, y) \models \gamma(\tau)$ for all $i\in[1,k]$. A regular resynchronizer is bounded if it is bounded by $k$ for some $k\in \N$.
	\end{definition}

Now, moving to simplified MSO resynchronizer in the present work is justified by the following Lemma:
	
	\begin{lemma}\label{lem:onlygamma}
If $R=(\alpha,\beta,\gamma,\delta)$ is a bounded extended MSO resynchronizer, then there exists a simplified MSO resynchronizer $R'$ that is also bounded, such that $\sem{R}\subseteq\sem{R'}$.
So if for two transducers $T_1$ and $T_2$ the relation $T_1\preceq T_2$ holds, as witnessed by a bounded extended resynchronizer, then it is also witnessed by a bounded simplified resynchronizer.
\end{lemma}
\begin{proof}
Let $m$ be the number of input parameters of $R$, and $\Theta$ its set of output-types.
The simplified resynchronizer $R'$ will use $m$ input parameters as well, and is defined by the formula 
$$\gamma'=\bigcup_{\tau\in\Theta} \gamma(\tau).$$

Let $k\in\N$ be such that $R$ is bounded by $k$. Let $K=k*|\Theta|$, we show that $R'$ is bounded by $K$.
Indeed, assume there are an input word $u$ labelled with input parameters $\bar{I}$, $K+1$ distinct positions $x_1,\dots, x_{K+1}$, and a position $y$, such that $(u,\bar{I},x_i,y)$ for all $i\in[1,K+1]$. Then by pigeonhole principle, there exists $\tau$ such that $(u,\bar{I},x_i,y)\gamma(\tau)$ is true for $k+1$ distinct values of $i$. This contradicts the fact that $R$ is bounded by $k$.

Finally, the fact that $\sem{R}\subseteq\sem{R'}$ is straightforward from the definition of $R'$: the presence of output parameters forcing $\gamma$ to use one of its disjuncts, and the addition of constraints $\alpha,\beta,\delta$, only restrict the semantics of a resynchronizer. Any pair of origin graphs $(\sigma,\sigma')$ accepted by $R$ is accepted by $R'$ as well, using the same input parameters as witness.
This means that if an extended resynchronizer $R=(\alpha,\beta,\gamma,\delta)$ witnesses $T_1\preceq T_2$, then $R'$ as defined here witnesses it as well.

\end{proof}

Therefore, as far as the relation $\preceq$ is concerned, we can assume that all bounded resynchronizers are in simplified form, and we do so throughout the paper.

\subsection{Proof of Lemma \ref{lem:transitive}}\label{ap:transitive}

We want to show that $\preceq$ is reflexive and transitive.

Let $T$ be a 2NT, we have $T\preceq T$, witnessed by the MSO resynchronizer $\gamma(x,y)= (x=y)$. This resynchronizer preserves the strict origin semantics, and is bounded by $1$. This shows reflexivity of $\preceq$.
\medskip

Let $T_1,T_2,T_3$ be 2NTs such that $T_1\preceq T_2 \preceq T_3$.
This means there exists $R_1,R_2$ bounded such that $T_1\subseteq R_1(T_2)$ and $T_2\subseteq R_2(T_3)$.
Let $m_1,\gamma_1$ (resp. $m_2,\gamma_2$) be the numbers of input parameters and MSO formula of $R_1$ (resp. $R_2$).
We define a resynchronizer $R$ with $m=m_1+m_2$ input parameters, by 
$$\gamma(\bar{I},x_3,x_1)=\exists x_2.\gamma_1(\tau_1)(\bar{I_1},x_2,x_1)\wedge\gamma_2(\bar{I_2},x_3,x_2),$$
where $\bar{I_1}$ (resp. $\bar{I_2}$) is obtained from $\bar{I}$ by restriction to the first $n_1$ (resp. last $n_2$) components.
The formula $\gamma$ guesses a valid position $x_2$ for the position of the origin according to $T_2$, and uses it to redirect the origin from $x_3$ to $x_1$ directly.

It remains to verify that $R$ is a witness that $T_1\preceq T_3$, i.e. that $T_1\subseteq R(T_3)$.
Let $\sigma_1=(u,v,\orig_1)\in\osem{T_1}$, we know from $T_1\subseteq R_1(T_2)$ that there exists $\sigma_2=(u,v,\orig_2)\in\osem{T_2}$ such that $(\sigma_2,\sigma_1)\in\sem{\gamma_1}$, witnessed by parameters $\bar{I_1}$. From $T_2\subseteq R_2(T_3)$, there exists $\sigma_3=(u,v,\orig_3)\in\osem{T_3}$ such that $(\sigma_3,\sigma_2)\in\sem{\gamma_2}$, witnessed by parameters $\bar{I_2}$. Let us show that $(\sigma_3,\sigma_1)\in\sem{R}$. Let $\bar{I}$ be the concatenation $\bar{I_1}\cdot\bar{I_2}$.
Let $x\in\dom(v)$ be an output position. We need to show that $(u,\bar{I},\orig_3(x),\orig_1(x))\models\gamma$. For $i\in\{1,2,3\}$ let $x_i=\orig_i(x)$.
We have $(u,\bar{I_1},x_2,x_1)\models\gamma_1$ and $(u,\bar{I_2},x_3,x_2)\models\gamma_2$, therefore, by definition of $\gamma$, we have $(u,\bar{I},x_3,x_1)\models\gamma$.
This concludes the proof of $T_1\subseteq R(T_3)$.

\subsection{Proof of Lemma \ref{lem:rk}}\label{ap:proof_rk}

Each input position $x$ that can be redirected to the right (resp. left) is labelled by some $\Rsource$ (resp. $\Lsource$). Notice that these labels are not exclusive, and a position $x$ can a priori have many such labels. However our construction ensures that every position $x$ has at most one  right label and one left label.
		
		We construct an algorithm that builds the input parameters $\Lsource, \Rsource$ such that it witnesses $(\sigma,\sigma') \in \sem{R_k}$. We will describe how to assign $\Rsource$ parameters, the left variant is symmetrical. The parameter variable $\Rsource$ starts with value $\emptyset$ for each $i\in[0,k-1]$, and will be filled with new positions during the run of the algorithm.
		
		Now let $\Rdist = \{x_1, \dots , x_n\}\subseteq \dom(u)$ be the set (indexed in increasing order) of positions $x$ such that there exists an output position $t$ with $\orig(t)=x$ and $\orig'(t)>x$, i.e. $\Rdist$ is the set of positions that can be redirected to the right. The algorithm makes a left to right pass of the input positions in $\Rdist$, starting at $x_1$. When treating $x_j\in\Rdist$ it does the following:
		
		\begin{enumerate}
			\item Set $\free = \{ i ~|~ \forall x\in\Rsource, x\textit{ does not traverse } x_j\}$.
			\item If $\free$ is empty, then output ``error'' and stop, otherwise let $\imin$ be the minimal element of $\free$, and add $x_j$ to  $\Rparam_{\imin}$.
		\end{enumerate}
		
	
		If the algorithm never outputs ``error'', then by construction these input parameters witness $(\sigma, \sigma')\in \sem{R_k}$. Indeed, if a position $x$ traverses a position $z$, the algorithm cannot give the same label $\Rsource$ to both $x$ and $z$.
		
Notice that in the algorithm, the set of free indexes is recomputed from scratch at every step. Equivalently, we could remember for each $i$ the rightmost redirection target $y_i$ of the position $s_i$ currently labelled by $\Rsource$, and free index $i$ when we reach position $y_i$.		
		
		We prove that ``error'' will never be output, under the $k$-traversal hypothesis on $(\sigma,\sigma')$. Assume for contradiction that at stage $j$, $\free$ is empty. This means that for all $i\in[0,k-1]$, there is a position $s_i\in\Rsource$ that traverses $x_j$. These $s_i$ are all distinct, since by construction an input position is only added to at most one input parameter $\Rsource$. This shows that position $x_j$ is traversed by $k$ positions strictly before $x_j$, and since it also traverses itself, we have a contradiction with the $k$-traversal assumption.

\subsection{Construction of domino tiles}\label{ap:domino}

A \emph{configuration} of $M$ is the data of a tape content, a state, and the position of the head on the tape. Such a configuration will be encoded by a word of $\Gamma^*$ of the form $u\cdot q a\cdot v\#$, with $u,v\in A^*$, $q\in Q$, and $a\in A$. The symbol $\#$ is used as a separator, allowing to concatenate configurations to form a computation history of $M$. When necessary, intermediary configurations are interleaved to add blank symbols at the extremity of the tape.

The word $u\cdot q a\cdot v\#$ encodes a tape $uav$, with a machine in state $q$ currently reading the marked letter $a$.

The full computation history of $M$ on empty input is a finite or infinite sequence of configurations, and can be encoded by a single word $\hist\in\Gamma^*\cup\Gamma^\omega$, obtained by concatenation of the encodings of the successive configurations.

We will now associate a finite set of \emph{tiles} $D_M$ to the machine $M$.
Each tile of $D_M$ is indexed by an integer $i$, and consists of a pair of words $(u_i,v_i)\in(\Gamma^*)^2$.

The set $D_M$ contains the following tiles:
\begin{itemize}
\item for every $a\in A\cup\{\#\}$, a copy tile $(a,a)$,
\item for every right moving transition $\delta(p,a)=(q,b,\toright)$, a right tile $(pa,bq)$,
\item for every $q\in Q$, a right expansion tile $(q\#,qB\#)$,
\item for every left moving transition $\delta(p,a)=(q,b,\toleft)$, and every letter $c\in a$, a left tile $(cpa,qcb)$, as well as a left expansion tile $(\#pa, \#qBb)$.
\end{itemize}

Notice that we omitted to include a start tile $(\eps,q_0\#)$ in $D_M$, as we will encode it explicitly in the reduction. Let $\Sigma\subseteq\N$ be the finite set of indexes of tiles from $D_M$.
In the classical proof of undecidability of the Post Correspondence Problem \cite{PCP}, these tiles are designed to simulate the run of $M$ as specified by Lemma \ref{lem:domino}.


\subsection{Undecidability of $\boundtape$}\label{ap:boundtape}

\begin{lemma}
		For a deterministic Turing Machine $M$ it is undecidable whether $M \in \boundtape$. 
	\end{lemma}
	\begin{proof}
		We reduce from the halting problem on an empty tape. Consider a deterministic Turing machine $M$, we build a new Turing machine $M'$ which simulates $M$ by writing the full computation history of $M$ on its tape. This new machine $M'$ halts if and only if the computation of $M$ halts. Moreover, $M'$ halts if and only if $M'\in\boundtape$, regardless of the tape usage of $M$. Therefore, we have that $M$ halts on empty input if and only if $M'\in\boundtape$, which is the wanted reduction.
	\end{proof}
	
%
%
%
%


%
%
%

\subsection{Undecidability results for rational transducers}\label{ap:rational}

\newcommand\dl[2]{{{#1} \choose {#2}}}


We recall here briefly the definition of rational resynchronizations for 1NTs. See \cite{filiot2016equivalence} for a full presentation.

The notion of origin graph is replaced here by \emph{interleaved word}: we assume the input alphabet $\Sigma$ and the output alphabet $\Gamma$ to be disjoint, and we represent the origin information of a pair $(u,v)\in\Sigma^*\times\Gamma^*$ by a word $w\in(\Sigma\cup\Gamma)^*$, such that when keeping only the letters from $\Sigma$ (resp. $\Gamma$) in $w$, we obtain the word $u$ (resp. $v$).
The origin of an output letter $v_i\in\Gamma$ is then given by the letter $u_j$ of $\Sigma$ immediately preceding it in $w$.

Thus, a resynchronization is now a set of pairs of interleaved words $(w,w')$, stating that the origins encoded by $w$ can be changed to those encoded by $w'$. Notice that the length of $w$ and $w'$ are always equal, so such a pair can be seen as a word on alphabet $(\Sigma\cup\Gamma)^2$

A resynchronization is \emph{rational} if it is a regular language on alphabet $(\Sigma\cup\Gamma)^2$.

\begin{example}
Let us recall the origin graphs from Example \ref{example:resynch:long}.
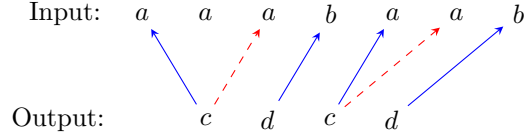
\begin{center}
		\begin{tikzpicture}[node distance=.4cm]
		
		\node (i1) {$a$};
		\node[right= of i1] (i2) {$a$};
		\node[right= of i2] (i3) {$a$};
		\node[right= of i3] (i4) {$b$};
		\node[right= of i4] (i5) {$a$};
		\node[right= of i5] (i6) {$a$};
		\node[right= of i6] (i7) {$b$};
		
		\node[below=1cm of i2] (o1) {$c$};
		\node[right= of o1] (o2) {$d$};
		\node[right= of o2] (o3) {$c$};
		\node[right= of o3] (o4) {$d$};

		\node[left=0.2cm of i1]	(labeli) {Input:};
		\node[left=1.0cm of o1] (labelo) {Output:};
		
		\draw [before] (o1) -- (i1);
		\draw [before] (o2) -- (i4);
		\draw [before] (o3) -- (i5);
		\draw [before] (o4) -- (i7);
		
		\draw [after] (o1) -- (i3);
		\draw [after] (o3) -- (i6);
		\end{tikzpicture}
\end{center}
		
The blue origin graph would be encoded by the interleaved word \textcolor{blue}{$acaabdacabd$}, and the red one by \textcolor{red}{$aaacbdaacbd$}.
So this particular resynchronization pair is represented by the pair of words $(\textcolor{blue}{acaabdacabd}, \textcolor{red}{aaacbdaacbd})$, that we can represent in columns to visualize the alphabet $(\Sigma\cup\Gamma)^2$:
$$
\left(
\begin{array}{c}
\textcolor{blue}{acaabdacabd}\\
\textcolor{red}{aaacbdaacbd}
\end{array}
\right)
$$
The resynchronizer $\Rblock$ from Example \ref{example:resynch:long} is rational, as witnessed by the following regular expression on alphabet $(\Sigma\cup\Gamma)^2$:
\newcommand{\eblock}{e_{\mathit{block}}}
\newcommand{\ebd}{e_{bd}}
$$
(\ebd)^*\left( \eblock(\ebd)^+\right)^*\eblock(\ebd)^*$$ 
$$\text{where }\ebd =\dl bb \dl dd\text{ and  }\eblock=\dl aa \left(\dl cc+ \dl ca \dl aa^*\dl ac\right)
$$

\end{example}

In particular it is shown in \cite{filiot2016equivalence} that the shift resynchronizations are rational (under the name \emph{bounded delay resynchronisers}).

\newcommand{\Tupalt}{T'_{up}}
\newcommand{\Tdownalt}{T'_{down}}

As mentioned in Section \ref{sec:undec}, since the shift resynchronizations are rational, and that any rational resynchronization is in particular bounded regular \cite[Theorem 3]{bose2019synthesis}, our reduction from Section \ref{sec:undec} can be used in particular as an alternative proof of undecidability of rational resynchronization synthesis, shown in \cite{bose2019synthesis} via one-counter automata. 
This means we directly obtain Corollary \ref{cor:rational}:

\rational*

We can further strengthen the result via Theorem \ref{thm:rationalundecidable}:

\ratreg*


 We prove this by a small modification of the construction of $\Tup$ from the undecidability proof in Section \ref{sec:undec}. We design $\Tupalt$  such that it either simulates $\Tup$, or outputs an arbitrary word with origin on the first input letter and then finishes. The transducer $\Tupalt$ is represented below:

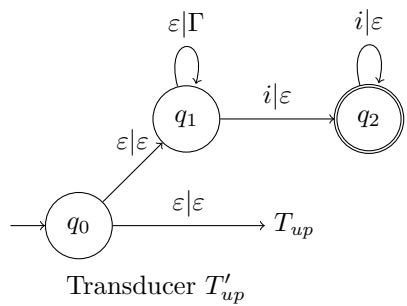
\begin{center}
	\begin{tikzpicture}[node distance = 2cm, initial text=]
	\node[initial,state] (init) {$q_0$};
	\node[state, above right of= init] (fail) {$q_1$};
	\node[state,accepting, right=1.5cm of fail] (last) {$q_2$};
	
	\node[right= of init] (success) {$\Tup$};
	
	\node[below right=.2cm and .7cm of init,anchor=north] {Transducer $\Tupalt$};
	
	\path[->]
	(init) edge node[above=.1cm] {$\eps|\eps$} (fail)
	(fail) edge[loop above] node[above] {$\eps|\Gamma$} (fail)
	(fail) edge node[above] {$i|\eps$} (last)
	(last) edge[loop above] node[above] {$i|\eps$} (last)
	(init) edge node[above] {$\eps|\eps$} (success)
	;
	\end{tikzpicture}
\end{center}
We have $\Tdown \preceq \Tupalt$, witnessed by the bounded resynchronizer $R$ defined by $\gamma(x,y)=first(x)$. In this resynchronizer, any origin pointing to the first input letter can be resynchronized to any input position. However, $R$ is not rational, and the existence of a rational resynchronizer witnessing $\Tdown\preceq\Tupalt$ reduces to $\boundtape$.

\begin{lemma}\label{lem:regtorat}
There exists a rational resynchronization $R$ such that $\Tdown \subseteq R(\Tupalt)$ if and only if $M\in \boundtape$.
\end{lemma}
\begin{proof}
Let $R$ be a rational resynchronizer such that for all graph $\sigma' \in \osem{\Tdown}$, there exists a graph $\sigma\in \osem{\Tupalt}$ such that $(\sigma,\sigma')\in \sem{R}$. We show that when the input word is long enough, the graph $\sigma$ corresponds to a run of $\Tupalt$ simulating $\Tup$. Assuming the contrary, we would obtain that the rational resynchronizer $R$ contains arbitrarily long pairs $p_n$ of the form $\left(\begin{array}{c}i_1v_1v_2\dots v_ni_2\dots i_n\\ i_1v_1i_2v_2\dots\dots i_nv_n\end{array}\right)$, with $i_j\in\Sigma$ and $v_j\in\Gamma$ for all $j$. Let $n$ be bigger than twice the number of states of a DFA $\A$ recognizing the rational resynchronization on alphabet $(\Sigma\cup\Gamma^2)$. We can pump a factor of length at least $2$ in the factor $\dl {v_1}{v_1}\dl {v_2}{i_2}\dl {v_3}{v_2}\dots\dl {v_n}{x}$ of the pair $p_n$. This way we can produce pairs of words accepted by $\A$, but whose projection to $\Gamma$ do not match, i.e. the output word is not the same before and after resynchronization. This means that $\A$ is not the automaton of a rational resynchronization, a contradiction. We obtained that there exists a constant $k\in\N$ such that for inputs longer than $k$, $\Tupalt$ behaves as $\Tup$. Thus the proof of Theorem \ref{thm:reduction} can now be used to show that a rational resynchronizer exists if and only if $M\in \boundtape$. This uses the fact that a $k$-shift resynchronization is rational, and that any rational resynchronization is in particular regular \cite[Theorem 3]{bose2019synthesis}.
\end{proof}

\end{document}